\documentclass[journal]{IEEEtran}

\usepackage{xcolor,soul,framed} 
\usepackage[justification=centering]{caption}

\colorlet{shadecolor}{yellow}
\usepackage[pdftex]{graphicx}
\graphicspath{{../pdf/}{../jpeg/}}
\DeclareGraphicsExtensions{.pdf,.jpeg,.png}

\usepackage{units}
\usepackage[cmex10]{amsmath}
\usepackage{array}
\usepackage{mdwmath}
\usepackage{mdwtab}
\usepackage{eqparbox}
\usepackage{url}
\usepackage{mathtools}
\usepackage{geometry}
\usepackage{hyperref}
\usepackage{xcolor}
\usepackage{cite}
\usepackage{multirow}
\usepackage{multicol}
\usepackage{booktabs}
\usepackage{subfig}
\usepackage[T1]{fontenc}
\usepackage{lipsum}
\usepackage[utf8]{inputenc}
\usepackage{array}
\usepackage{lettrine}
\usepackage{authblk}
\usepackage{graphicx}
\usepackage{stfloats}
\usepackage{amsmath}
\usepackage{epsfig}
\usepackage{amssymb}
\usepackage{caption}
\usepackage{comment}
\usepackage[utf8]{inputenc}
\usepackage[mathletters]{ucs}
\usepackage{tablefootnote}
\usepackage{algcompatible}
\usepackage{algorithm}
\usepackage{algpseudocode}
\usepackage{setspace}
\usepackage{bbm}
\usepackage[english]{babel}
\usepackage{amsthm}

\newtheorem{theorem}{Theorem}
\newtheorem{proposition}[theorem]{Proposition}

\usepackage[skip=2pt,font=footnotesize]{caption}

\begin{document}
   \title{ \huge Robust Beamforming Design for Fairness-Aware Energy Efficiency Maximization in RIS-Assisted mmWave Communications} \vspace{-0.3cm}
    	\newgeometry {top=25.4mm,left=19.1mm, right= 19.1mm,bottom =19.1mm}%
\author{Ahmed Magbool,~\IEEEmembership{Graduate Student Member,~IEEE,} 
Vaibhav Kumar,~\IEEEmembership{Member,~IEEE,} 
\\  \vspace{-0.3cm} and Mark F. Flanagan,~\IEEEmembership{Senior Member,~IEEE}\thanks{This publication has emanated from research conducted with the financial
support of Science Foundation Ireland under Grant Number 13/RC/2077\_P2. \par
Ahmed Magbool and Mark F. Flanagan are with the School of Electrical and Electronic Engineering, University College Dublin, Belfield, Dublin 4, Dublin 4, D04 V1W8, Ireland (e-mail: ahmed.magbool@ucdconnect.ie, mark.flanagan@ieee.org). \par
Vaibhav Kumar was with the School of Electrical and Electronic Engineering, University College Dublin, Dublin 4, D04 V1W8, Ireland. He is now with Engineering Division, New York University Abu Dhabi (NYUAD), Abu Dhabi 129188, UAE (e-mail: vaibhav.kumar@ieee.org).} \vspace{-0.9cm}}

\maketitle

\begin{abstract}
Users in millimeter-wave (mmWave) systems often exhibit diverse channel strengths, which can negatively impact user fairness in resource allocation. Moreover, exact channel state information (CSI) may not be available at the transmitter, rendering suboptimal resource allocation. In this paper, we address these issues within the context of energy efficiency maximization in RIS-assisted mmWave systems. We first derive a tractable lower bound on the achievable sum rate, taking into account CSI errors. Subsequently, we formulate the optimization problem, targeting maximizing the system energy efficiency while maintaining a minimum Jain's fairness index controlled by a tunable design parameter. The optimization problem is very challenging due to the coupling of the optimization variables in the objective function and the fairness constraint, as well as the existence of non-convex equality and fractional constraints. To solve the optimization problem, we employ the penalty dual decomposition method, together with a projected gradient ascent based alternating optimization procedure. Simulation results demonstrate that the proposed algorithm can achieve an optimal energy efficiency for a prescribed Jain's fairness index. In addition, adjusting the fairness design parameter can yield a favorable trade-off between energy efficiency and user fairness compared to methods that exclusively focus on optimizing one of these metrics.
\end{abstract}
\vspace{-0.6cm}
\IEEEpeerreviewmaketitle
\section{Introduction} \label{sec:intro}
The millimeter-wave (mmWave) band, extending from \unit[28]{GHz} to \unit[100]{GHz}, has been exploited in existing fifth-generation (5G) wireless networks and is planned for use in the upcoming sixth-generation (6G) wireless networks~\cite{6G_survey,THz_NOMA}. However, signals at such high frequencies experience extreme propagation conditions, limiting the system's capabilities~\cite{mmwave_comm_survey}.

The utilization of reconfigurable intelligent surfaces (RISs) is regarded as a promising solution to address the challenge of obscured line-of-sight (LoS) in mmWave communications. RISs can establish a programmable propagation environment by manipulating the phase of the incident signal through nearly passive reflecting elements. Consequently, higher spectral efficiency can be attained with reduced power consumption compared to the corresponding non-RIS-aided systems~\cite{RIS_mot1, RIS_mot3, RIS_mot4}.

Sustainability is a crucial aspect of 6G systems, aiming to achieve an energy efficiency (EE) target of \unit[1]{Tbit/Joule}~\cite{6G_goals, massive_mimoo, sus_EE}. With that in mind, the escalating power consumption associated with massive data transmission within 6G networks has led researchers to focus on EE maximization for 6G networks~\cite{EE_opt1, EE_opt2}. A milestone work in EE optimization for RIS-aided systems was presented in~\cite{EE_opt3}, offering two solutions to determine the optimal beamforming design. The first solution used a gradient-based approach, while the second relied on a fractional-programming-based alternating optimization (AO) method. In~\cite{EE_opt4}, a genetic approach was proposed, employing a covariance matrix adaptation evolution strategy and Dinkelbach's method to maximize the system EE for terahertz (THz) systems. Optimal beamformer design, coupled with an asymptotic channel capacity characterization under hardware impairments, was addressed in~\cite{EE_opt5}. The problem of secrecy EE maximization in a multiple-input multiple-output multiple eavesdropper (MIMOME) channel was tackled in~\cite{V_paper}, where a stationary solution was obtained using a penalty-based approach. In addition, a long short-term memory (LSTM)-based solution was explored to design an energy-efficient RIS-aided system in~\cite{EE_opt6}.

 However, the aforementioned works overlooked two crucial issues. The first is the need to ensure an adequate degree of fairness among users. Since different users may have distinct quality-of-service (QoS) requirements, \textit{proportional fairness} should be imposed, ensuring that the \textit{weighted rates} of all users together satisfy a fairness measure~\cite{w_rate}. Ignoring this aspect leads to allocating minimal resources to users with weaker channels, as their contribution to the EE objective function is very minor~\cite{ICC_paper}. This problem is more pronounced in the mmWave band due to the severe path loss, resulting in users close to the transmitter having much stronger channel gains than those farther away from it. While incorporating rate constraints into the optimization problem (e.g., by setting a specific threshold for each user's rate) can impose some level of fairness, simulation results in~\cite{EE_opt3} indicate that these threshold constraints can severely affect the EE. Moreover, depending on the choice of these rate thresholds, a feasible solution to the optimization problem may not exist. Second, while the above-referenced papers assumed perfect channel state information (CSI) knowledge at the transmitter, this assumption does not hold in practice.
 
In this paper, we target maximizing the system EE while ensuring user fairness for mmWave multiuser systems with an array-of-subarrays (AoSA) hybrid beamforming architecture. We choose this architecture because it is more energy-efficient compared to the a fully-connected digital beamforming architecture~\cite{AoSA_motivation, HB}. The proposed method is robust to CSI imperfections and relies only on the assumption of bounded CSI error. The main contributions of this paper are summarized as follows:
 \begin{itemize}
     \item For the first time in the literature, we present a solution to the problem of EE maximization for a transmitter with AoSA architecture employing hybrid beamforming in the context of RIS-assisted mmWave communications. Our solution is sufficiently general to cover for the case of multiple-antenna receivers capable of decoding multiple streams simultaneously. Unlike previous works on this topic, which often assume a special precoder (e.g., the zero-forcing (ZF) precoder in~\cite{EE_opt3}), our approach addresses a more general case with the presence of multi-user interference (MUI). 
     \item We address the practical scenario where perfect CSI is not accessible at the transmitter. To model the channel estimation error, we adopt the widely-used bounded-error model~\cite{CEE11,CEE2}, which constrains the estimation error of the cascaded channel within a sphere of a specific radius. Employing this bounded error model, we derive a tractable lower bound on the achievable sum rate through the triangle and Cauchy–Schwarz inequalities.
     \item We formulate the EE maximization problem with a constraint on users' fairness among users, characterized by Jain's fairness index. The fairness constraint imposes specific proportions between the users' weighted rates, leading to an improved system EE compared to the standard formulation found in the literature, which typically involves minimum rate thresholds.
     \item To address the challenging non-convex optimization problem, we employ a penalty dual decomposition (PDD) method by incorporating the fairness constraint into the objective function. Subsequently, we utilize a projected gradient based AO algorithm to determine the optimal digital precoder, analog precoder, RIS reflection matrix, and receive combiners that maximize the system EE.
     \item We present comprehensive numerical results to demonstrate the performance of the proposed method. Our findings indicate that the proposed method can achieve superior trade-offs between EE and fairness compared to the state-of-the-art EE and fairness maximization techniques discussed in~\cite{EE_opt3} and~\cite{fairness_3}.
 \end{itemize}

A related problem was previously addressed in our work~\cite{ICC_paper}. However, that study was limited to single-antenna receivers and MUI-free reception, while the current work extends the scope to include the AoSA architecture, multiple-antenna receivers, and the presence of MUI. Additionally, in~\cite{ICC_paper}, perfect CSI knowledge is assumed to be available at the transmitter, whereas this work emphasizes robust system design in the presence of channel estimation errors. While~\cite{ICC_paper} aimed at maximizing both EE and user fairness, this paper concentrates on EE maximization with a fairness constraint. Given the nature of the more complicated optimization problem in this work, we utilize the projected gradient ascent and the PDD methods, which are fast and numerically efficient approaches for large-scale systems. The proposed method does not require specific convex solvers such as CVX or CVXPY. 
 
 The remainder of the paper is organized as follows. Section~\ref{sec:sys_model} presents the system model. Section~\ref{sec:prob_form} formulates the EE maximization problem with fairness constraint. Section~\ref{sec:sol_S1} presents a detailed description of the proposed penalty-based AO framework. Simulation results and associated discussion are presented in Section~\ref{sec:sim}. Finally, the paper is concluded in Section~\ref{sec:conc}.

 \textit{Notations:} Bold lowercase and uppercase letters denote vectors and matrices, respectively. $|\cdot|$ represents the magnitude of a complex number (for a complex vector, it is assumed to operate element-wise). $[\mathbf{a}]_i$ denotes the $i$-th element of the vector $\mathbf{a}$. $||\mathbf{\cdot} ||_2$ and $||\mathbf{\cdot} ||_\mathsf{F}$ represent the Euclidean vector norm and the Frobenius matrix norm, respectively. $\mathbf{(\cdot)}^\mathsf{T}$, $\mathbf{(\cdot)}^\mathsf{H}$ and $\mathbf{(\cdot)}^{-1}$ denote the matrix transpose, matrix conjugate transpose, and matrix inverse, respectively. $\mathbf{I}$ represents the identity matrix. $\boldsymbol{1}$ denotes a column vector all of whose elements are equal to one. $\text{diag}(a_1,\dots,a_N)$ denotes a diagonal matrix with the elements $a_1,\dots,a_N$ on the main diagonal and zeros elsewhere, while $\text{blkdiag}(\mathbf{a}_1,\dots,\mathbf{a}_N)$ represents a block diagonal matrix with the vectors $\mathbf{a}_1,\dots,\mathbf{a}_N$ on the main block diagonal and zeros elsewhere. $\mathbb{C}$ indicates the set of complex numbers, and $j = \sqrt{-1}$ is the imaginary unit. $\mathbb{E}\{ \cdot \}$ stands for the expectation operator. $\mathcal{CN}(\mathbf{0},\mathbf{B})$ represents a complex Gaussian random vector with a mean of $\mathbf{0}$ and a covariance matrix of $\mathbf{B}$. $\nabla_{\mathbf a} f(\cdot)$ represents the complex-valued gradient of the real-valued function $f(\cdot)$ with respect to $\mathbf a$. $\odot$ and $\otimes$ denote the Khatri-Rao product and the Kronecker product, respectively. $\mathcal{O}(\cdot)$ represents the big-O notation to denote the computational complexity of an algorithm.

\vspace{-0.2cm}

\section{System Model} \label{sec:sys_model} 
\vspace{-0.1cm}
We consider a multi-user downlink mmWave system in which a base-station (BS) serves $K$ users as shown in Fig.~\ref{fig:sys_model}. The BS needs to transmit $L$ data symbols to each user, resulting in a total of $KL$ data streams. The BS employs hybrid digital and analog precoding, where the analog precoder comprises $M$ radio-frequency (RF) chains, each connected to a sub-array (SA) of $N_\mathsf{T}$ antenna elements (AEs). We assume that the direct paths between the BS and the user equipment (UEs) are blocked and the communication is achieved via an RIS with $N_\text{RIS}$ reflecting elements. At the receiver side, each UE is assumed to have $N_\mathsf{R}$ AEs, and employs a digital combiner to decode $L$ data streams, simultaneously\footnote{Due to constraints related to signal processing capability, power consumption, and hardware design, the number of streams for each UE (i.e., $L$) is typically small. In particular, the following conditions must hold: $LK \leq M \leq N_\mathsf{T}$ and $L\leq N_\mathsf{R}$.}\footnote{This is a general system model and can be simplified into several special cases. For instance, the case of single-antenna receivers can be derived by setting  $N_\mathsf{R} = L = 1$, while $N_\mathsf{T} = 1$  represents the scenario of fully digital antennas at the BS.}. In the following subsections, we describe the channel model, the signal model, and the channel estimation error model.
\begin{figure}
         \centering
         \includegraphics[width=1\columnwidth]{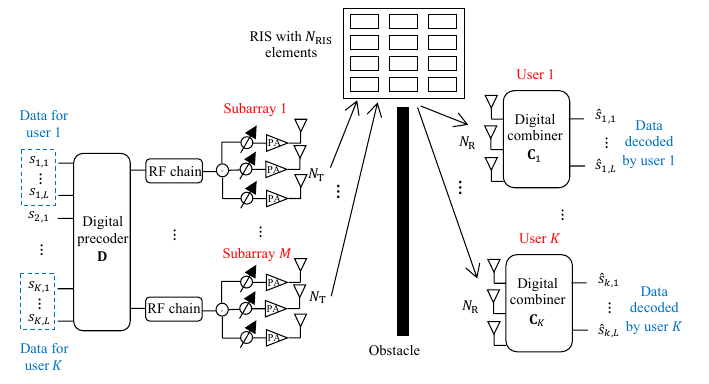}
        \caption{System model for the proposed RIS-assisted mmWave system.}
        \label{fig:sys_model}
        \vspace{-0.6cm}
\end{figure}

\vspace{-0.4cm}
\subsection{Channel Model}
We adopt the widely used far-field Saleh-Valenzuela (SV) channel model~\cite{ch_model_1,ch_model_2,ch_model_3} to represent the BS-RIS and RIS-UE channels. We can write the BS-RIS channel as
\begin{equation}
    \mathbf{H}_{\mathsf{T}} =  \sum_{\ell = 0}^{L_\mathsf{T}-1} \alpha_{\mathsf{T},\ell} \mathbf{a}_\text{RIS} (\vartheta_{\mathsf{R},\ell}, \varphi_{\mathsf{R},\ell}) \mathbf{a}_\text{BS} (\phi_{\mathsf{T},\ell})^\mathsf{H},
    \label{sv_channel_model_1}
\end{equation}
 where $L_\mathsf{T}$ is the number of paths between the BS and the RIS with $\ell=0$ representing the LoS path and $\ell=1,\dots,L_\mathsf{T}-1$ representing the non-line-of-sight (NLoS) paths, and $ \alpha_{\mathsf{T},\ell}$ is the complex path gain of the $\ell$-th path. In addition, $\mathbf{a}_\text{BS} (\phi_{\mathsf{T},\ell})$ is the beam steering vector at the BS, which is a function of the angle of departure (AoD) from BS $\phi_{\mathsf{T},\ell}$ of the $\ell$-th path from the BS. Assuming that the distance between two adjacent AEs is $\lambda_{\mathrm c}/2$, with $\lambda_{\mathrm c}$ indicating the carrier wavelength, and the distance between the first antennas in two adjacent SAs is $d_\text{SA}\lambda_{\mathrm c}/2$, the array response of the $n$-th AE at the $m$-th SA at the BS can be written as
\begin{equation}
  a_{\text{BS}}^{m,n} (\phi_{\mathsf{T},\ell}) =  \exp \Big( { j \pi \big((n-1) + (m-1) d_{\text{SA}} \big) \cos (\phi_{\mathsf{T},\ell}}) \Big) .
\end{equation}

 Also, $\mathbf{a}_\text{RIS} (\vartheta_{\mathsf{R},\ell}, \varphi_{\mathsf{R},\ell})$ is the beam steering vector of the RIS, which is a function of the azimuth and elevation angles of arrival (AoA) to the RIS,  $\vartheta_{\mathsf{R},\ell}$ and $\varphi_{\mathsf{R},\ell}$, respectively. Assuming that the horizontal and vertical distances between two RIS elements is $\lambda_{\mathrm c}/2$, the element in the $n_x$-th row and $n_y$-th column in the uniform planar array (UPA) RIS have an array response of
\begin{equation}
\begin{split}
     a_{\text{RIS}}^{n_x,n_y} (\vartheta_{\mathsf{R},\ell}, \varphi_{\mathsf{R},\ell}) & =  \exp \Big( j \pi \big(  \varpi_{n_x}(\vartheta_{\mathsf{R},\ell},\varphi_{\mathsf{R},\ell}) \\
     & + \varpi_{n_y}(\vartheta_{\mathsf{R},\ell},\varphi_{\mathsf{R},\ell}) \big)  \Big), 
    \end{split}
\end{equation}
\hspace{-0.4cm} where $ \varpi_{n_x}(\vartheta_{\mathsf{R},\ell},\varphi_{\mathsf{R},\ell}) = (n_x - 1) \cos (\vartheta_{\mathsf{R},\ell}) \sin (\varphi_{\mathsf{R},\ell})$ and $ \varpi_{n_y}(\vartheta_{\mathsf{R},\ell},\varphi_{\mathsf{R},\ell}) = (n_y - 1) \sin (\vartheta_{\mathsf{R},\ell}) \sin (\varphi_{\mathsf{R},\ell})$.

 Similarly, the channel between the RIS and the $k$-th UE is
\begin{equation}
     \mathbf{H}_{\mathsf{R},k} \hspace{-0.1cm} = \sum_{\ell = 0}^{L_{\mathsf{R},k}-1} \hspace{-0.25cm} \alpha_{\mathsf{R},k,\ell} \mathbf{a}_{\text{UE},k} (\Xi_{\mathsf{R},k,\ell}) \mathbf{a}_\text{RIS} (\vartheta_{\mathsf{T},\ell}, \varphi_{\mathsf{T},\ell})^\mathsf{H},
    \label{sv_channel_model_2}
\end{equation}
 where $L_{\mathsf{R},k}$ is the number of paths between the RIS and the $k$-th UE, $\alpha_{\mathsf{R},k,\ell} $ is the path gain of the $\ell$-th path between the RIS and the $k$-th UE, and $\vartheta_{\mathsf{T},\ell}$ and $\varphi_{\mathsf{T},\ell}$ are the azimuth and elevation AoD from the RIS, respectively. Additionally, $\mathbf{a}_{\text{UE},k} (\Xi_{\mathsf{R},k,\ell})$ is array response vector for the $k$-th UE, which can be written as (assuming $\lambda_{\mathrm c}/2$ spacing between AEs)
\begin{equation}
    \mathbf{a}_{\text{UE},k} (\Xi_{\mathsf{R},k,\ell}) =  \big[1,\dots ,\exp \big(j \pi (N_\mathsf{R}-1) \cos(\Xi_{\mathsf{R},k,\ell}) \big) \big],
\end{equation}
where $\Xi_{\mathsf{R},k,\ell}$ is the AoA of the $\ell$-th path to the $k$-th UE.

\subsection{Signal Model}
The signal transmitted by the BS can be expressed as
\begin{equation}
    \mathbf{x} = \mathbf{A} \mathbf{D} \mathbf{s},
\end{equation}
where $\mathbf{s} = [s_{1,1},\dots,s_{K,L}]^\mathsf{T} \in \mathbb{C}^{KL \times 1}$ is a vector containing the transmitted symbols, with $s_{k,\ell}$ representing the $\ell$-th symbol to be transmitted to the $k$-th user, where $k \in \{1, 2, ..., K\}$ and $\ell \in \{1, 2, ..., L\}$, and $\mathbb{E}\{ \mathbf{s}\mathbf{s}^\mathsf{H} \} = \mathbf{I}$. Also, $\mathbf{D} = [\mathbf{d}_{1,1}, \dots, \mathbf{d}_{K,L}]  \in \mathbb{C}^{M \times KL}$ is the digital precoder satisfying the power budget constraint, i.e., $|| \mathbf{D} ||_\mathsf{F}^2 \leq P_\text{max}$, with $ P_\text{max}$ denoting the total transmit power budget and $\mathbf{d}_{k,\ell}  \in \mathbb{C}^{M \times 1}$ denoting the digital precoder for the $\ell$-th data stream intended for the $k$-th user. In addition, $\mathbf{A} = \text{blkdiag}(\mathbf{a}_{1} ,\dots, \mathbf{a}_{M} ) \in \mathbb{C}^{MN_\mathsf{T} \times M}$ is the analog precoder, with $\mathbf{a}_{m} \in \mathbb{C}^{N_\mathsf{T} \times 1} $ representing the analog precoder employed by the $m$-th RF chain, with elements obeying the constant modulus (CM) constraint, i.e., $|[ \mathbf{a}_{m} ]_i| = 1 / \sqrt{N_\mathsf{T}}$.

At the other end, the post-combining received signal by the $k$-th user to decode the $\ell$-th data stream is
\begin{equation}
    y_{k,\ell} = \mathbf{c}_{k,\ell}^\mathsf{H} \mathbf{G}_k (\boldsymbol{\Theta}) \mathbf{x} + \mathbf{c}_{k,\ell}^\mathsf{H} \mathbf{n}_k,
    \label{eq:rec_signal}
\end{equation}
where $\mathbf{c}_{k,\ell}  \in \mathbb{C}^{N_\mathsf{R} \times 1}$ is the combiner employed by the $k$-th user to decode the $\ell$-th data stream that satisfies the receive power budget constraint $\sum_{\ell = 1}^{L} || \mathbf{c}_{k,\ell} ||_2^2 = P_{\mathsf{R},k}$, and $ \mathbf{n}_k \sim \mathcal{CN} (\mathbf{0}, \sigma^2 \mathbf{I})$ is the complex additive white Gaussian  noise (AWGN). Furthermore, $\mathbf{G}_k (\boldsymbol{\Theta}) \in \mathbb{C}^{N_\mathsf{R} \times MN_\mathsf{T}}$ denotes the overall channel matrix for the $k$-th user, which can be written as
\begin{equation}
     \mathbf{G}_k (\boldsymbol{\Theta}) =   \mathbf{H}_{\mathsf{R},k} \boldsymbol{\Theta} \mathbf{H}_{\mathsf{T}},
    \label{eq:ccc}
\end{equation}
where $\boldsymbol{\Theta}= \text{diag}(\theta_1,\dots,\theta_{N_\text{RIS}}) \in \mathbb{C}^{N_\text{RIS} \times N_\text{RIS} }$ is the RIS reflection coefficient with $| \theta_n| = 1$ for all $ n \in \{1,\dots, N_\text{RIS} \}$. However, in practice, due to channel estimation errors, perfect CSI will in general not be available at the transmitter. In the next subsection, we introduce the channel estimation error model.

\vspace{-0.7cm}
\subsection{Channel Estimation Error Model}
\vspace{-0.1cm}
To characterize the channel estimation error model, we express the channel of the $k$-th user as
\begin{equation}
    \mathbf{G}_k (\boldsymbol{\Theta}) = \hat{\mathbf{G}}_k (\boldsymbol{\Theta}) + \boldsymbol{\Lambda}_k (\boldsymbol{\Theta}),
    \label{eq:sbs}
\end{equation}
where $\hat{\mathbf{G}}_k (\boldsymbol{\Theta})$ represents the estimated channel matrix at the BS, and $\boldsymbol{\Lambda}_k (\boldsymbol{\Theta})$ denotes the error matrix resulting from imperfect estimation of $\mathbf{G}_k (\boldsymbol{\Theta})$. 

As RISs are inherently passive, they lack the capability of directly transmitting or processing pilots. Consequently, most channel estimation algorithms focus on estimating the cascaded channel~\cite{CE_sur}, which can be obtained by vectorizing~\eqref{eq:sbs} as
\begin{equation}
    \text{vec} (\mathbf{G}_k (\boldsymbol{\Theta})) = \underbrace{(\mathbf{H}_{\mathsf{T}}^\mathsf{T} \odot \mathbf{H}_{\mathsf{R},k})}_{\mathbf{H}_k} \boldsymbol{\theta} = (\hat{\mathbf{H}}_k + \boldsymbol{\Delta}_k) \boldsymbol{\theta},
    \label{eq:vec_channel}
\end{equation}
where $\boldsymbol{\theta} = \text{diag} (\boldsymbol{\Theta})$, $\hat{\mathbf{H}}_k$ is the estimated version of $\mathbf{H}_k$, and $\boldsymbol{\Delta}_k$ is the error matrix incurred during the estimation of $\mathbf{H}_k$. The relationship between $\boldsymbol{\Delta}_k$ and $\boldsymbol{\Lambda}_k (\boldsymbol{\Theta})$ in~\eqref{eq:sbs} can be expressed as $\text{vec} \big(\boldsymbol{\Lambda}_k (\boldsymbol{\Theta})\big) = \boldsymbol{\Delta}_k \boldsymbol{\theta}$.

In this paper, we assume that the channel estimation error is bounded within a sphere of radius $\delta_k$. This bounded error model is widely adopted in the literature due to its effectiveness in characterizing quantization errors which lie within a bounded region~\cite{CEE11,CEE2}.

\section{Problem Formulation}  \label{sec:prob_form}
In this section, we formulate the problem of maximizing the system EE while adhering to a fairness constraint. We begin by expressing the rate of the $\ell$-th stream decoded by the $k$-th user as
\begin{equation}
\begin{split}
     & R_{k,\ell}  (\mathbf{D}, \mathbf{A}, \boldsymbol{\Theta}, \mathbf{c}_{k,\ell} ) =  \log_2 \bigg(  1 + \\
       &  \frac{| \mathbf{c}_{k,\ell}^\mathsf{H} \mathbf{G}_k (\boldsymbol{\Theta}) \mathbf{A} \mathbf{d}_{k,\ell }|^2}{\underset{[i,v] \neq [k,\ell]}{\sum_{i=1}^K \sum_{ {v=1}}^L}| \mathbf{c}_{k,\ell}^\mathsf{H} \mathbf{G}_k (\boldsymbol{\Theta}) \mathbf{A} \mathbf{d}_{i,v } |^2 + \mathbb{E}\{ |\mathbf{c}_{k,\ell}^\mathsf{H} \mathbf{n}_k|^2 \}} \bigg).
     \end{split}
    \label{eq:Rk}
\end{equation}

In addition, the total power consumption of the system can be expressed as (c.f.~\cite{EE_opt3})
\begin{equation}
\begin{split}
    P_{\text{tot}} (\mathbf{D})  & = \underbrace{P_\text{BS}+ \xi|| \mathbf{D}||_\mathsf{F}^2 + MN_\mathsf{T} P_{\text{RF,T}}}_{\text{BS power consumption}} \\ &  + \underbrace{N_\text{RIS} P_{\theta}}_{\text{RIS power consumption}}  +\underbrace{ \sum_{k=1}^K ( P_{\text{UE},k}  + P_{\text{R},k} )}_{\text{UEs power consumption}},
    \end{split}
    \label{p_cons}
\end{equation}
where $P_\text{BS}$ and $P_{\text{UE},k}$ are the total hardware static power consumption at the BS and at the $k$-th UE, respectively, $\xi$ is the power amplification factor at the BS, $P_{\text{RF,T}}$ and $P_{\theta}$ are the power consumption of each phase shifter at the BS and at the RIS, respectively. \newpage

We can then express the system EE as the sum rate divided by the total power consumption, i.e.,
\begin{equation}
    \eta (\mathbf{D},\mathbf{A}, \boldsymbol{\Theta}, \mathbf{C}) = \frac{R_\text{sum} (\mathbf{D},\mathbf{A}, \boldsymbol{\Theta}, \mathbf{C})}{P_{\text{tot}} (\mathbf{D})},
    \label{eq:EE}
\end{equation}
where $R_\text{sum}  (\mathbf{D},\mathbf{A}, \boldsymbol{\Theta}, \mathbf{C}) \hspace{-0.1cm} \triangleq \hspace{-0.1cm} 
 \sum_{K=1}^K \sum_{ \ell=1}^L \hspace{-0.1cm} 
 R_{k,\ell} (\mathbf{D}, \mathbf{A}, \boldsymbol{\Theta}, \mathbf{c}_{k,\ell} )$ and $\mathbf{C} \triangleq [\mathbf{c}_{1,1},\dots, \mathbf{c}_{K,L}]$. The goal is to allocate the available resources to maximize the system EE. However, as the channel estimation error matrices $\boldsymbol{\Delta}_1,\dots,\boldsymbol{\Delta}_K$ are not accessible at the transmitter, we introduce the following proposition.
\begin{proposition} \label{prop:1}
The rate of the $\ell$-th data stream decoded by the $k$-th user can be lower bounded as
\begin{equation}
\begin{split}
     & R_{k,\ell} (\mathbf{D},\mathbf{A}, \boldsymbol{\Theta}, \mathbf{c}_{k,\ell})  \geq  \underline{R}_{k,\ell} (\mathbf{D},\mathbf{A}, \boldsymbol{\Theta}, \tilde{\mathbf{c}}_{k,\ell})  \triangleq \log_2 \bigg( 1 + \\
     & \frac{ \varrho_k | \tilde{\mathbf{c}}_{k,\ell}^\mathsf{H} \hat{\mathbf{G}}_{k} (\boldsymbol{\Theta}) \mathbf{A} \mathbf{d}_{k,\ell }|^2 }{\underset{[i,v] \neq [k,\ell]}{\sum_{i=1}^K \sum_{ v=1}^L} \big( | \tilde{\mathbf{c}}_{k,\ell}^\mathsf{H} \hat{\mathbf{G}}_{k} (\boldsymbol{\Theta}) \mathbf{A} \mathbf{d}_{i,v }| + \delta_k \sqrt{N_\mathsf{RIS}} || \mathbf{d}_{i,v }||_2 \big)^2 + \sigma^2 } \bigg),
     \end{split}
     \label{eq:ny_eq}
\end{equation}
where $\varrho_k \triangleq \big(1-\frac{\delta_k}{|| \hat{\mathbf{H}}_k||_\mathsf{F}}\big)^2$ and $\tilde{\mathbf{c}}_{k,\ell} \triangleq \mathbf{c}_{k,\ell} / || \mathbf{c}_{k,\ell}||_2$.
\end{proposition}
\begin{proof}
See Appendix~\ref{AppndA}.
\end{proof}

Using Proposition~\ref{prop:1}, we can obtain a lower bound on the system EE as
\begin{equation}
    \eta (\mathbf{D},\mathbf{A}, \boldsymbol{\Theta}, \mathbf{C}) \geq \underline{\eta} (\mathbf{D},\mathbf{A}, \boldsymbol{\Theta}, \tilde{\mathbf{C}}) \triangleq \frac{\underline{R}_\text{sum} (\mathbf{D},\mathbf{A}, \boldsymbol{\Theta}, \tilde{\mathbf{C}})}{P_{\text{tot}} (\mathbf{D})},
    \label{eq:EE_lb}
\end{equation}
where $\underline{R}_\text{sum}  (\mathbf{D},\mathbf{A}, \boldsymbol{\Theta}, \tilde{\mathbf{C}}) \hspace{-0.1cm} \triangleq \hspace{-0.1cm} \sum_{k=1}^K \hspace{-0.1cm} \sum_{\ell=1}^{L} \underline{R}_{k,\ell}  (\mathbf{D},\mathbf{A}, \boldsymbol{\Theta}, \tilde{\mathbf{c}}_{k,\ell})$ and $\tilde{\mathbf{C}} \triangleq [\tilde{\mathbf{c}}_{1,1},\dots, \tilde{\mathbf{c}}_{K,L}]$. 

While it is crucial from a system design perspective to maintain high EE by maximizing~\eqref{eq:EE_lb}, ensuring acceptable QoS remains a primary concern for end-users. However, users in mmWave communication systems often exhibit considerable variations in their channel strengths due to the severe path loss~\cite{mmwave_comm_survey}. In such cases, weaker users, unable to significantly contribute to the EE objective function~\eqref{eq:EE_lb}, "are likely to be treated as unwanted interference to stronger users, resulting in a limited resource allocation for the weaker ones. To address this issue, we aim to maximize the system EE while ensuring a certain level of fairness among users. This objective is pursued by tackling the following optimization problem:
\begin{subequations}
\begin{align}
           & \max_{  \mathbf{D},  \mathbf{A},\boldsymbol{\Theta}, \tilde{\mathbf{C}} }   \ \ \   \underline{\eta} (\mathbf{D},\mathbf{A}, \boldsymbol{\Theta}, \tilde{\mathbf{C}}) \label{optm_MO_obj} \\
        & \text{s.t.} \ \  \rho \leq \frac{\big(  \sum_{k=1}^K  \frac{1}{w_k} \sum_{\ell=1}^L \underline{R}_{k,\ell} (\mathbf{D},\mathbf{A}, \boldsymbol{\Theta}, \tilde{\mathbf{c}}_{k,\ell}) \big)^2}{K  \sum_{k=1}^K  \big( \frac{1}{w_k} \sum_{\ell=1}^L \underline{R}_{k,\ell} (\mathbf{D},\mathbf{A}, \boldsymbol{\Theta}, \tilde{\mathbf{c}}_{k,\ell}) \big)^2 } \label{F_const} \\ 
        & || \mathbf{D} ||_\mathsf{F}^2 \leq P_\text{max}, \label{Pow_con} \\
        & | \theta_n | = 1, \  \forall n \in \{ 1,\dots,N_\text{RIS} \}, \label{ps_con_EEf}\\
        & | [\mathbf{a}_{m}]_i | = \frac{1}{\sqrt{N_\mathsf{T}}}, \  \forall m \in \{ 1,\dots,M \}, \  \forall i \in \{ 1,\dots,N_\mathsf{T} \},\label{CM_BS} \\
        & || \tilde{\mathbf{c}}_{k,\ell} ||_2 = 1, \ \forall k \in \{1,\dots,K \}, \ \forall \ell \in \{1,\dots,L \},
     \label{lb_constraint}
\end{align} 
\label{eq:optm_MO} 
\end{subequations}
\hspace{-0.2cm} where $w_k$ is a \textit{weight} assigned to the $k$-th user by the system provider, representing the QoS for that user. Additionally, $\rho \in [1/K, 1]$ serves as a design parameter incorporated to manage the Jain's fairness index of the system. Jain's fairness index is a widely recognized metric that quantifies the fairness level among proportional rates in a system, ranging between $1/K$ and 1. A Jain's fairness index of 1 indicates perfect fairness, implying that all users have the same proportional rate~\cite{sim5}.

Solving~\eqref{eq:optm_MO} poses a significant challenge due to the coupling between the optimization variables in the objective function and the fairness constraint. Moreover, handling the non-convex equality constraints~\eqref{ps_con_EEf}, ~\eqref{CM_BS} and~\eqref{lb_constraint} is challenging. In the following section, we introduce projected gradient ascent based AO procedures to address this optimization problem.

\vspace{-0.3cm}

\section{Proposed Solution} \label{sec:sol_S1}
\vspace{-0.1cm}
In this section, we present a penalty-based method to get rid of the complicated fairness constraint, which otherwise is very challenging to handle due to the fractional form, and also due to the coupling between the design variables. We then apply an AO-based approach to update each of the design variables in a one-by-one fashion. 

The general framework for updating each variable comprises three main steps:

\textbf{1) Obtaining the augmented Lagrangian function:} First, to address the challenging fairness constraint~\eqref{F_const}, the PDD method~\cite{pen} is employed. The PDD method operates by augmenting challenging constraints with penalty into the objective function, and ensures that both the original and the transformed problems share the same stationary solution. For the problem at hand, to obtain the penalty function, we first rewrite the constraint~\eqref{F_const} as
    \begin{equation}
      \rho K  \sum_{k=1}^K  \Big( \frac{1}{w_k} \sum_{\ell=1}^L \underline{R}_{k,\ell} (\boldsymbol{\chi}) \Big)^2  - \Big(  \sum_{k=1}^K  \frac{1}{w_k} \sum_{\ell=1}^L \underline{R}_{k,\ell} (\boldsymbol{\chi}) \Big)^2 \leq 0, 
     \label{eq:con_re}
\end{equation}
where $\boldsymbol{\chi} \in \{ \mathbf{D},\mathbf{A}, \boldsymbol{\Theta}, \tilde{\mathbf{C}}\}$ is the relevant optimization variable\footnote{In this paper, having some arguments omitted in a multi-variable function indicates that their values are being held fixed.}. Based on~\eqref{eq:con_re}, we define the following penalty function:
\begin{equation}
\begin{split}
    \mathcal{G}(\boldsymbol{\chi}, \mu)   & \triangleq \rho K  \sum_{k=1}^K  \Big( \frac{1}{w_k} \sum_{\ell=1}^L \underline{R}_{k,\ell} (\boldsymbol{\chi} ) \Big)^2 \\
    & - \Big(  \sum_{k=1}^K  \frac{1}{w_k} \sum_{\ell=1}^L \underline{R}_{k,\ell} (\boldsymbol{\chi} ) \Big)^2 + \mu,
    \end{split}
    \label{eq:trans_cons}
\end{equation}
where $\mu \geq 0$ is a slack variable used to eliminate the penalty function whenever the condition~\eqref{eq:con_re} is satisfied.

Then, the PDD method converts the following constrained optimization problem:
\begin{subequations}
\begin{align}
            \max_{ \boldsymbol{\chi}}  \ \ \ & \underline{\eta} (\boldsymbol{\chi}) \\
        \text{s.t.} \ \ \ & \eqref{eq:con_re}, 
\end{align} 
\end{subequations} 
into the following form:
\begin{subequations}
\begin{align}
            \max_{ \boldsymbol{\chi}, \mu \geq 0}  \ \ \ & \mathcal{H}_{\gamma, \omega} (\boldsymbol{\chi}, \mu),
\end{align} 
        \label{eq:new_form}
\end{subequations}
 \hspace{-0.2cm}where $\mathcal{H}_{\gamma, \omega} (\boldsymbol{\chi},\mu)$ is the augmented Lagrangian function in the form
\begin{equation}
    \mathcal{H}_{\gamma, \omega} (\boldsymbol{\chi}, \mu)  \triangleq \underline{\eta} (\boldsymbol{\chi}) -  \gamma \mathcal{G}(\boldsymbol{\chi}, \mu) 
     - \frac{1}{2 \omega}   \mathcal{G}^2(\boldsymbol{\chi}, \mu),
    \label{eq:aug_lag}
\end{equation}
where $\gamma$ is the Lagrange multiplier and $\omega > 0$ is a penalty parameter~\cite{pen}.

For fixed values of $\gamma$ and $\omega$, the optimal solution of the optimization problem~\eqref{eq:new_form} is obtained using a projection-based AO procedure, as will be demonstrated later in this section. Once convergence is achieved for specific values of $\gamma$ and $\omega$, they can be updated as~\cite{pen}:
 \begin{equation}
     \gamma \leftarrow \gamma + \frac{1}{\omega} \mathcal{H}_{\gamma, \omega} 
    (\mathbf{D},\mathbf{A}, \boldsymbol{\Theta}, \tilde{\mathbf{C}},\mu),
 \end{equation}
 and 
  \begin{equation}
     \omega \leftarrow \psi \omega,
 \end{equation}
where $\psi \in (0,1]$ is a design parameter that controls the speed of attenuation of the constraint violation.

\textbf{2) Gradient ascent update:} To update the relevant variable $\boldsymbol{\chi}$, the gradient of the augmented Lagrangian function with respect to the optimization variable (i.e., $\nabla_{\boldsymbol{\chi}} \mathcal{H}_{\gamma, \omega} (\boldsymbol{\chi}) $)  is obtained. Then, the variable is updated via
\begin{equation}
    \bar{\boldsymbol{\chi}}^{(t+1)} = \boldsymbol{\chi}^{(t)} + \alpha_{\boldsymbol{\chi}}^{(t)}  \nabla_{\boldsymbol{\chi}}\mathcal{H}_{\gamma, \omega} (\boldsymbol{\chi}^{(t)}),
\end{equation}
where $t$ is the iteration number, $\alpha_{\boldsymbol{\chi}}^{(t)}>0$ is the step size for updating $\boldsymbol{\chi}$ at the $t$-th iteration, and $\bar{\boldsymbol{\chi}}^{(t+1)}$ and $\boldsymbol{\chi}^{(t)}$ are the values of $ \boldsymbol{\chi}$ at the $(t+1)$-th and the $t$-th iterations of the gradient ascent algorithm, respectively.

\textbf{3) Projection onto the feasible set:} The updated value of the variable $\boldsymbol {\chi}$, i.e., $\bar{\boldsymbol {\chi}}^{(t+1)}$ may not lie in its feasible set. Therefore, we project the updated variable orthogonally onto the feasible set using the projection function:
\begin{equation}
    \boldsymbol{\chi}^{(t+1)} = \mathcal{T}_{\boldsymbol{\chi}}(\bar{\boldsymbol{\chi}}^{(t+1)}).
\end{equation}

We detail the procedure of updating each variable in the following subsections.
\vspace{-0.4cm}
\subsection{Digital Precoder}
\vspace{-0.1cm}
We first assume that all variables other than the digital precoder are fixed (i.e., the values of $\mathbf{A},\boldsymbol{\Theta}, \tilde{\mathbf{C}}$ and $\mu$ are held fixed). We can then write the EE maximization problem with respect to the digital precoder as
 \begin{subequations}
\begin{align}
            \max_{   \mathbf{D} }  \ \ \  & \underline{\eta} (\mathbf{D}) \label{eq:OF_EE_F_BB} \\
        \text{s.t.} \ \ \ & \eqref{F_const}, \ \eqref{Pow_con}. 
        \vspace{-0.2cm}
\end{align} 
\label{eq:optm_EE_F_BB} 
\end{subequations} 
The augmented Lagrangian function after incorporating the constraint~\eqref{F_const} into~\eqref{eq:OF_EE_F_BB} is
\begin{equation}
       \mathcal{H}_{\gamma, \omega}(\mathbf{D})  = \underline{\eta} (\mathbf{D}) -  \gamma \mathcal{G}(\mathbf{D}) 
     - \frac{1}{2 \omega}  \mathcal{G}^2(\mathbf{D}) ,
     \label{eq:grad_ALF_D}
\end{equation}

To find the gradient of~\eqref{eq:grad_ALF_D} with repect to $\mathbf{D}$, we begin by re-expressing $\underline{R}_{k,\ell} (\mathbf{D})$ in \eqref{eq:ny_eq} as
\begin{equation}
\begin{split}
        & \underline{R}_{k,\ell} (\mathbf{D})   = \log_2 \Big( \varrho_k | \mathbf{b}_{k,\ell}^\mathsf{H} \mathbf{d}_{k,\ell }|^2  \\ 
 & +
 \underset{[i,v]\neq [k,\ell]}{\sum_{i=1}^K \sum_{v=1}^{L}} \big(| \mathbf{b}_{k,\ell}^\mathsf{H} \mathbf{d}_{i,v}| + \delta_k \sqrt{N_\text{RIS}} || \mathbf{d}_{i,v }||_2 \big)^2 + \sigma^2 \Big)  \\
        &   - \log_2 \Big( \underset{[i,v]\neq [k,\ell]}{\sum_{i=1}^K \sum_{v=1}^{L}} \big(| \mathbf{b}_{k,\ell}^\mathsf{H} \mathbf{d}_{i,v}| + \delta_k \sqrt{N_\text{RIS}} || \mathbf{d}_{i,v }||_2\big)^2 + \sigma^2 \Big),
     \end{split}
     \label{eq:der_F_BB_EE}
\end{equation}
where $\mathbf{b}_{k,\ell} = \mathbf{A}^\mathsf{H} \hat{\mathbf{G}}_k^\mathsf{H} (\boldsymbol{\Theta}) \tilde{\mathbf{c}}_{k,\ell}$.

We can obtain the gradient of $\underline{R}_{k,\ell} (\mathbf{d}_{m,u}) $ with respect to $\mathbf{d}_{m,u}$ by considering two cases. First, when $[m,u] = [k,\ell]$, we have the gradient expression shown in~\eqref{eq:grad_c1}.
\begin{figure*} [b]
  \hrulefill
  \begin{equation}
  \nabla_{\mathbf{d}_{m,u} } \underline{R}_{k,\ell} (\mathbf{D}) \Big|_{[m,u] = [k,\ell]} = \frac{2}{\ln 2} \ \frac{\varrho_k\mathbf{b}_{k,\ell} \mathbf{b}_{k,\ell}^\mathsf{H} \mathbf{d}_{k,\ell}}{\varrho_k | \mathbf{b}_{k,\ell}^\mathsf{H} \mathbf{d}_{k,\ell } |^2 +
 \underset{[i,v] \neq [k,\ell] }{\sum_{i=1}^K \sum_{v=1}^{L}} (| \mathbf{b}_{k,\ell}^\mathsf{H} \mathbf{d}_{i,v}| +  \delta_k \sqrt{N_\text{RIS}}  || \mathbf{d}_{i,v }||_2)^2 + \sigma^2} .
 \label{eq:grad_c1}
  \end{equation}
\end{figure*}
On the other hands, when $[m,u] \neq [k,\ell]$, we can obtain the gradient expression shown in~\eqref{eq:grad_c2}.
\begin{figure*} [b]
  \hrulefill
  \begin{equation}
\begin{split}
    \nabla_{\mathbf{d}_{m,u}} \underline{R}_{k,\ell} (\mathbf{D}) \Big|_{[m,u] \neq [k,\ell]} = \frac{2}{\ln 2} \bigg[ & \frac{ \big(1 + \frac{ \delta_k \sqrt{N_\text{RIS}} ||  \mathbf{d}_{m,u } ||_2}{| \mathbf{b}_{k,\ell}^\mathsf{H} \mathbf{d}_{m,u }| } \big) \mathbf{b}_{k,\ell} \mathbf{b}_{k,\ell}^\mathsf{H} + \big( \delta_k^2 N_\text{RIS} + \frac{  \delta_k \sqrt{N_\text{RIS}} | \mathbf{b}_{k,\ell}^\mathsf{H} \mathbf{d}_{m,u } |}{ || \mathbf{d}_{m,u } ||_2} \big) \mathbf{I} }{\varrho_k | \mathbf{b}_{k,\ell}^\mathsf{H} \mathbf{d}_{k,\ell } |^2 +
 \underset{[i,v] \neq [k,\ell] }{\sum_{i=1}^K \sum_{v=1}^{L}} (| \mathbf{b}_{k,\ell}^\mathsf{H} \mathbf{d}_{i,v}| +  \delta_k \sqrt{N_\text{RIS}}  || \mathbf{d}_{i,v }||_2)^2 + \sigma^2} \\
 & - \frac{ \big(1 + \frac{ \delta_k \sqrt{N_\text{RIS}} ||  \mathbf{d}_{m,u } ||_2}{| \mathbf{b}_{k,\ell}^\mathsf{H} \mathbf{d}_{m,u }| } \big) \mathbf{b}_{k,\ell} \mathbf{b}_{k,\ell}^\mathsf{H} + \big( \delta_k^2 N_\text{RIS} + \frac{  \delta_k \sqrt{N_\text{RIS}} | \mathbf{b}_{k,\ell}^\mathsf{H} \mathbf{d}_{m,u } |}{ || \mathbf{d}_{m,u } ||_2} \big) \mathbf{I} }{
 \underset{[i,v] \neq [k,\ell] }{\sum_{i=1}^K \sum_{v=1}^{L}} (| \mathbf{b}_{k,\ell}^\mathsf{H} \mathbf{d}_{i,v}| +  \delta_k \sqrt{N_\text{RIS}}  || \mathbf{d}_{i,v }||_2)^2 + \sigma^2} \bigg] \mathbf{d}_{m,u}.
 \label{eq:grad_c2}
 \end{split}
  \end{equation}
\end{figure*}
Thus, the gradient of $\underline{R}_{k,\ell} (\mathbf{D}) $ with respect to $\mathbf{d}_{m,u}$ can be expressed as
\begin{equation}
    \nabla_{ \mathbf{d}_{m,u }} \underline{R}_{k,\ell} (\mathbf{D})= 
      \begin{cases}
       \eqref{eq:grad_c1}, & \text{if} \ [m,u] = [k,\ell],\\
       \eqref{eq:grad_c2}, & \text{if} \  [m,u] \neq [k,\ell].
    \end{cases} 
    \label{eq:prop2}
\end{equation}
The gradient of $\underline{R}_{k,\ell} (\mathbf{D} )$ with respect to $ \mathbf{D}$ can then be written as
\begin{equation}
      \nabla_{  \mathbf{D}} \underline{R}_{k,\ell} ( \mathbf{D}) = \big[ \nabla_{ \mathbf{d}_{1,1 }} \underline{R}_{k,\ell} (\mathbf{D}),\dots, \nabla_{ \mathbf{d}_{K,L }} \underline{R}_{k,\ell} (\mathbf{D}) \big],
\end{equation}
and that of  $\underline{R}_\text{sum} (\mathbf{D})$ with respect to $ \mathbf{D}$ is
\begin{equation}
      \nabla_{  \mathbf{D}} \underline{R}_\text{sum} ( \mathbf{D}) = \sum_{k=1}^K \sum_{\ell=1}^{L} \nabla_{  \mathbf{D}} \underline{R}_{k,\ell} ( \mathbf{D}).
      \label{eq:cc_3}
\end{equation}

We can also obtain the gradient of $P_\text{tot} (\mathbf{D})$ with respect to $\mathbf{D}$ as
\begin{equation}
    \nabla_{\mathbf{D}} P_\text{tot}  (\mathbf{D}) = 2 \xi \mathbf{D}.
    \label{eq:cc_4}
\end{equation}
Then the gradient of $ \underline{\eta}(\mathbf{D})$ with respect to $\mathbf{D}$ is given by
\begin{equation}
    \nabla_{\mathbf{D}}  \underline{\eta}  (\mathbf{D}) = \frac{ P_\text{tot}  (\mathbf{D}) \nabla_{\mathbf{D}} \underline{R}_\text{sum} (\mathbf{D}) - \underline{R}_\text{sum} (\mathbf{D}) \nabla_{\mathbf{D}} P_\text{tot}  (\mathbf{D}) }{P_\text{tot}^2 (\mathbf{D})}.
    \label{eq:f_BBk_update_EE}
\end{equation}

We also need to find the gradient of $\mathcal{G} (\mathbf{D})$ with respect to $\mathbf{D}$, which can be written as
\begin{equation}
\begin{split}
    & \nabla_\mathbf{D} \mathcal{G}(\mathbf{D})   = 2 \rho K  \Big( \sum_{k=1}^K   \frac{1}{w_k^2} \sum_{\ell=1}^L \underline{R}_{k,\ell} (\mathbf{D} ) \sum_{m=1}^L \nabla_\mathbf{D} \underline{R}_{k,m} (\mathbf{D} ) \Big)  \\
     & - 2 \Big( \sum_{k=1}^K  \frac{1}{w_k}  \sum_{\ell=1}^L \underline{R}_{k,\ell} (\mathbf{D} ) \Big) \Big(  \sum_{i=1}^K  \frac{1}{w_i}  \sum_{m=1}^L \nabla_\mathbf{D} \underline{R}_{i,m} (\mathbf{D} ) \Big).
    \end{split}
\end{equation}
Thus, the gradient of $\mathcal{H}_{\gamma, \omega} 
    (\mathbf{D})$ with respect to $\mathbf{D}$ is
\begin{equation}
       \nabla_{\mathbf{D}} \mathcal{H}_{\gamma, \omega} (\mathbf{D})  = \nabla_{\mathbf{D}} \underline{\eta} (\mathbf{D}) -  \gamma \nabla_{\mathbf{D}}  \mathcal{G}(\mathbf{D}) 
     - \frac{1}{ \omega}  \mathcal{G}(\mathbf{D}) \nabla_{\mathbf{D}}  \mathcal{G}(\mathbf{D}),
\end{equation}
and the update of $\mathbf{D}$ should follow
\begin{equation}
    \bar{\mathbf{D}}^{(t+1)} = \mathbf{D}^{(t)} + \alpha_{\mathbf{D}}^{(t)}  \nabla_{\mathbf{D}} \mathcal{H}_{\gamma, \omega} (\mathbf{D}^{(t)}).
    \label{eq:F_BB_update_f}
\end{equation}
 
  Next, we note that the projection of the updated variable $\bar{\mathbf{D}}^{(t+1)}$ onto the constraint \eqref{Pow_con} is given by
 \begin{equation}
     \mathbf{D}^{(t+1)} = \min \bigg( 1,\frac{\sqrt{P_\text{max}}}{|| \bar{\mathbf{D}}^{(t+1)} ||_\mathsf{F}} \bigg) \bar{\mathbf{D}}^{(t+1)}.
     \label{eq:proj_f_bb}
 \end{equation}

 \subsection{Analog Precoder}
 Assuming that the values of $\mathbf{D}, \boldsymbol{\Theta}$ and $\tilde{\mathbf{C}}$ are fixed, the EE optimization for the analog precoder is
 \begin{subequations}
\begin{align}
            \max_{  \mathbf{A} }  \ \ \  & \underline{\eta}(\mathbf{A}), \label{eq:33a} \\
        \text{s.t.} \ \ \ & \eqref{F_const}, \ \eqref{CM_BS}. 
\end{align} 
\label{eq:optm_EE_F_RF} 
\end{subequations} 

We first express the augmented Lagrangian function of~\eqref{eq:33a} when augmenting the constraint~\eqref{F_const} into the objective function as
\begin{equation}
    \mathcal{H}_{\gamma, \omega} (\mathbf{A})  = \frac{1}{P_\text{tot}} \underline{R}_\text{sum} (\mathbf{A}) -  \gamma \mathcal{G}(\mathbf{A}) 
     - \frac{1}{2 \omega} \mathcal{G}^2(\mathbf{A})  .
     \label{eq:aug_Lag_R_sum}
\end{equation}

 To derive the gradient expression of the sum rate with respect to $ \mathbf{A}$, we begin by expressing the $k$-th user's hybrid precoding vector for the $\ell$-th stream as
\begin{equation}
    \mathbf{f}_{k,\ell} =  \mathbf{A} \mathbf{d}_{k,\ell} = \big[ [\mathbf{d}_{k,\ell}]_1  \mathbf{a}_{1}^\mathsf{T},\dots, [\mathbf{d}_{k,\ell}]_M  \mathbf{a}_{M}^\mathsf{T} \big]^\mathsf{T} = \Tilde{\mathbf{D}}_{k,\ell} \mathbf{a},
\end{equation}
 where $\Tilde{\mathbf{D}}_{k,\ell} =  \text{diag} ( \big[ [\mathbf{d}_{k,\ell}]_1 \boldsymbol{1}^\mathsf{T} ,\dots,  [\mathbf{d}_{k,\ell}]_M \boldsymbol{1}^\mathsf{T} \big]^\mathsf{T})$, and $ \mathbf{a} = [\mathbf{a}_{1}^\mathsf{T},\dots,\mathbf{a}_{M}^\mathsf{T}]^\mathsf{T}$. Subsequently, the rate of the $k$-th user decoding the $\ell$-th stream can be written as
 \begin{equation}
 \begin{split}
    & \underline{R}_{k,\ell} (\mathbf{a})    =  \log_2 \Big( \varrho_k | (\mathbf{e}_{k,\ell}^{k,\ell})^\mathsf{H} \mathbf{a} |^2  \\
    &+  \underset{[i,v] \neq [k,\ell]}{
    \sum_{i=1}^K \sum_{v=1}^{L}} (|(\mathbf{e}_{k,\ell}^{i,v})^\mathsf{H} \mathbf{a} | +  \delta_k \sqrt{N_\text{RIS}} || \mathbf{d}_{i,v}||_2)^2 + \sigma^2 \Big)\\
    & - \log_2 \Big( \underset{[i,v] \neq [k,\ell]}{
    \sum_{i=1}^K \sum_{v=1}^{L}} (|(\mathbf{e}_{k,\ell}^{i,v})^\mathsf{H} \mathbf{a} | +  \delta_k \sqrt{N_\text{RIS}} || \mathbf{d}_{i,v}||_2)^2 + \sigma^2 \Big),
     \end{split}
     \label{eq:rate_frf}
\end{equation}
where $\mathbf{e}_{k,\ell}^{i,v} = \Tilde{\mathbf{D}}_{i,v}^\mathsf{H} \hat{\mathbf{G}}^\mathsf{H}(\boldsymbol{\Theta}) \tilde{\mathbf{c}}_{k,\ell} $.

Then, the gradient of $\underline{R}_{k,\ell} (\mathbf{a})$ with respect to $\mathbf{a}$ is given in~\eqref{eq:grad_rs_a}.
\begin{figure*} [b]
  \hrulefill
  \begin{equation}
  \begin{split}
    \nabla \underline{R}_{k,\ell} (\mathbf{a}) = \frac{2}{\ln 2} \bigg[ & \frac{   \varrho_k \mathbf{e}_{k,\ell}^{k,\ell}(\mathbf{e}_{k,\ell}^{k,\ell})^\mathsf{H}  + \underset{[i,v] \neq [k,\ell]}{
    \sum_{i=1}^K \sum_{v=1}^{L}}  \big(1 + \frac{ \delta_k \sqrt{N_\text{RIS}} || \mathbf{d}_{i,v}||_2}{| (\mathbf{e}_{k,\ell}^{k,\ell})^\mathsf{H} \mathbf{a} |   } \big) \mathbf{e}_{k,\ell}^{i,v} (\mathbf{e}_{k,\ell}^{i,v})^\mathsf{H} }{\varrho_k | (\mathbf{e}_{k,\ell}^{k,\ell})^\mathsf{H} \mathbf{a} |^2 +  \underset{[i,v] \neq [k,\ell]}{
    \sum_{i=1}^K \sum_{v=1}^{L}}  (|(\mathbf{e}_{k,\ell}^{i,v})^\mathsf{H} \mathbf{a} | +  \delta_k \sqrt{N_\text{RIS}} || \mathbf{d}_{i,v}||_2)^2 + \sigma^2} \\
    &- \frac{\underset{[i,v] \neq [k,\ell]}{
    \sum_{i=1}^K \sum_{v=1}^{L}}  \big(1 + \frac{ \delta_k \sqrt{N_\text{RIS}} || \mathbf{d}_{i,v}||_2}{| (\mathbf{e}_{k,\ell}^{k,\ell})^\mathsf{H} \mathbf{a} |   } \big) \mathbf{e}_{k,\ell}^{i,v} (\mathbf{e}_{k,\ell}^{i,v})^\mathsf{H} }{\underset{[i,v] \neq [k,\ell]}{
    \sum_{i=1}^K \sum_{v=1}^{L}}  (|(\mathbf{e}_{k,\ell}^{i,v})^\mathsf{H} \mathbf{a} | +  \delta_k \sqrt{N_\text{RIS}} || \mathbf{d}_{i,v}||_2)^2 + \sigma^2}  \bigg] \mathbf{a}.
    \end{split}
    \label{eq:grad_rs_a}
\end{equation}
\end{figure*}
Thus the gradient of~\eqref{eq:33a} with respect to $\mathbf{a}$ is given by
\begin{equation}
    \nabla_{\mathbf{a}} \underline{R}_\text{sum} (\mathbf{a}) =  \sum_{k=1}^K \sum_{\ell=1}^{L} \nabla_{\mathbf{a}} \underline{R}_{k,\ell} (\mathbf{a}).
\end{equation}

Similarly, we can obtain the gradient of $\mathcal{G}(\mathbf{a})$ with respect to $\mathbf{a}$ as
\begin{equation}
\begin{split}
    & \nabla_\mathbf{a} \mathcal{G}(\mathbf{a})   = 2 \rho K  \Big( \sum_{k=1}^K   \frac{1}{w_k^2} \sum_{\ell=1}^L \underline{R}_{k,\ell} (\mathbf{a} ) \sum_{m=1}^L \nabla_\mathbf{a} \underline{R}_{k,m} (\mathbf{a} ) \Big)  \\
     & - 2 \Big( \sum_{k=1}^K  \frac{1}{w_k}  \sum_{\ell=1}^L \underline{R}_{k,\ell} (\mathbf{a} ) \Big) \Big(  \sum_{i=1}^K  \frac{1}{w_i}  \sum_{m=1}^L \nabla_\mathbf{a} \underline{R}_{i,m} (\mathbf{a} ) \Big).
    \end{split}
\end{equation}
Therefore, the gradient of the augmented Lagrangian function~\eqref{eq:aug_Lag_R_sum} with respect to $\mathbf{a}$ can be written as
\begin{equation}
    \nabla_\mathbf{a} \mathcal{H}_{\gamma, \omega}  (\mathbf{a})  = \frac{\nabla_\mathbf{a} \underline{R}_\text{sum} (\mathbf{a})}{P_\text{tot}} -  \gamma \nabla_\mathbf{a} \mathcal{G}(\mathbf{a}) 
     - \frac{1}{ \omega} \ \mathcal{G}(\mathbf{a}) \nabla_\mathbf{a} \mathcal{G}(\mathbf{a}) .
\end{equation}
Hence, the update of $\mathbf{a}$ should follow
\begin{equation}
    \bar{\mathbf{a}}^{(t+1)} = \mathbf{a}^{(t)} + \alpha_{\mathbf{a}}^{(t)}  \nabla_\mathbf{a} \mathcal{H}_{\gamma, \omega}   (\mathbf{a}^{(t)}),
     \label{eq:F_RF_update_EE}
\end{equation}
\hspace{-0.4cm} and the projection of $\bar{\mathbf{a}}^{(t+1)}$ onto the constraint~\eqref{CM_BS} can be performed as
\begin{equation}
    \mathbf{a}^{(t+1)} = \frac{1}{\sqrt{N_\mathsf{T}}} \text{diag}(|\bar{\mathbf{a}}^{(t+1)}|)^{-1} \bar{\mathbf{a}}^{(t+1)} .
    \label{eq:proj_f_RF}
\end{equation}

\subsection{RIS Reflection Matrix}
Next, assuming that the values of $\mathbf{D},\mathbf{A}$ and $\tilde{\mathbf{C}}$ are fixed, the RIS reflection matrix optimization sub-problem can be stated as 
\begin{subequations}
\begin{align}
           \max_{ \boldsymbol{\Theta} }  \ \ \ & \underline{\eta}(\boldsymbol{\Theta}), \\
        \text{s.t.} \ \ \ & \eqref{F_const}, \ \eqref{ps_con_EEf},
\end{align} 
\label{eq:optm_EE_RIS} 
\end{subequations} 
with the following augmented Lagrangian function:
\begin{equation}
    \mathcal{H}_{\gamma, \omega}(\boldsymbol{\Theta})  = \frac{1}{P_\text{tot}}\underline{R}_\text{sum} (\boldsymbol{\Theta}) -  \gamma \mathcal{G}(\boldsymbol{\Theta}) 
     - \frac{1}{2 \omega}   \mathcal{G}^2(\boldsymbol{\Theta}) .
     \label{eq:aug_Lag_theta}
\end{equation}

To find the gradient of $\mathcal{H}_{\gamma, \omega}  (\boldsymbol{\Theta})$ with respect to $\boldsymbol{\Theta}$, we use the fact that the vectorization of a scalar term is the term itself to re-express the terms $ \tilde{\mathbf{c}}_{k,\ell}^\mathsf{H} \hat{\mathbf{G}}_{k} (\boldsymbol{\Theta}) \mathbf{A} \mathbf{d}_{i,v }$ as
\begin{equation}
     \tilde{\mathbf{c}}_{k,\ell}^\mathsf{H} \hat{\mathbf{G}}_{k} (\boldsymbol{\Theta}) \mathbf{A} \mathbf{d}_{i,v } =  \underbrace{ \big( (\mathbf{A} \mathbf{d}_{i,v })^\mathsf{T} \otimes \tilde{\mathbf{c}}_{k,\ell}^\mathsf{H} \big) \hat{\mathbf{H}}_k }_{(\mathbf{f}_{k,\ell}^{i,v})^\mathsf{H}} \boldsymbol{\theta},
     \label{eq:new_exrp}
\end{equation}
 which can be obtained by applying the identity $\text{vec} (\boldsymbol{\Omega}_1 \boldsymbol{\Omega}_2 \boldsymbol{\Omega}_3) = (\boldsymbol{\Omega}_3^\mathsf{T} \otimes \boldsymbol{\Omega}_1) \text{vec} (\boldsymbol{\Omega}_2)$, which holds for any general matrices $\boldsymbol{\Omega}_1$, $\boldsymbol{\Omega}_2$ and $\boldsymbol{\Omega}_3$ with appropriate dimensions~\cite{MCB}.

Thus, we can write the rate of the $k$-th user decoding the $\ell$-th stream as
\begin{equation}
 \begin{split}
    & \underline{R}_{k,\ell} (\boldsymbol{\theta}) =  \log_2 \Big( \varrho_k |(\mathbf{f}_{k,\ell}^{k,\ell})^\mathsf{H} \boldsymbol{\theta}|^2  \\
    & + \underset{[i,v] \neq [k,\ell]}{\sum_{i=1}^K \sum_{v=1}^L } \big( |(\mathbf{f}_{k,\ell}^{i,v})^\mathsf{H} \boldsymbol{\theta}| + \delta_k \sqrt{N_\text{RIS}} || \mathbf{d}_{i,v}||_2 \big)^2 + \sigma^2 \Big) \\
    & - \log_2 \Big(  \underset{[i,v] \neq [k,\ell]}{\sum_{i=1}^K \sum_{v=1}^L } \big( |(\mathbf{f}_{k,\ell}^{i,v})^\mathsf{H} \boldsymbol{\theta}| + \delta_k \sqrt{N_\text{RIS}} || \mathbf{d}_{i,v}||_2 \big)^2 + \sigma^2 \Big). 
    \end{split}
\end{equation}
Then, the gradient of $ \underline{R}_{k,\ell} (\boldsymbol{\theta})$ with respect to $\boldsymbol{\theta}$ is given by~\eqref{eq:grad_drs_a}, shown below, 
\begin{figure*} [b]
  \hrulefill
 \begin{equation}
\begin{split}
    \nabla_{\boldsymbol{\theta}} \underline{R}_{k,\ell} (\boldsymbol{\theta}) = \frac{2}{\ln 2} \bigg[ & \frac{\varrho_k \mathbf{f}_{k,\ell}^{k,\ell} (\mathbf{f}_{k,\ell}^{k,\ell})^\mathsf{H} + \underset{[i,v] \neq [k,\ell]}{\sum_{i=1}^K \sum_{v=1}^L } \big(1 + \frac{  \delta_k \sqrt{N_\text{RIS}}|| \mathbf{d}_{i,v}||_2}{|(\mathbf{f}_{k,\ell}^{i,v})^\mathsf{H} \boldsymbol{\theta}|   } \big) \mathbf{f}_{k,\ell}^{i,v} (\mathbf{f}_{k,\ell}^{i,v})^\mathsf{H} }{\varrho_k |(\mathbf{f}_{k,\ell}^{k,\ell})^\mathsf{H} \boldsymbol{\theta}|^2 + \underset{[i,v] \neq [k,\ell]}{\sum_{i=1}^K \sum_{v=1}^L } \big( |(\mathbf{f}_{k,\ell}^{i,v})^\mathsf{H} \boldsymbol{\theta}| + \delta_k \sqrt{N_\text{RIS}} || \mathbf{d}_{i,v}||_2 \big)^2 + \sigma^2} \\
    &- \frac{\underset{[i,v] \neq [k,\ell]}{\sum_{i=1}^K \sum_{v=1}^L } \big(1 + \frac{ \delta_k \sqrt{N_\text{RIS}}|| \mathbf{d}_{i,v}||_2}{|(\mathbf{f}_{k,\ell}^{i,v})^\mathsf{H} \boldsymbol{\theta}|   } \big) \mathbf{f}_{k,\ell}^{i,v} (\mathbf{f}_{k,\ell}^{i,v})^\mathsf{H} }{ \underset{[i,v] \neq [k,\ell]}{\sum_{i=1}^K \sum_{v=1}^L } \big( |(\mathbf{f}_{k,\ell}^{i,v})^\mathsf{H} \boldsymbol{\theta}| + \delta_k \sqrt{N_\text{RIS}} || \mathbf{d}_{i,v}||_2 \big)^2 + \sigma^2}  \bigg] \boldsymbol{\theta}.
    \end{split}
     \label{eq:grad_drs_a}
\end{equation}
\end{figure*}
and that of $ \underline{R}_\text{sum} (\boldsymbol{\theta})$ is
\begin{equation}
    \nabla_{\boldsymbol{\theta}} \underline{R}_\text{sum} (\boldsymbol{\theta}) = \sum_{k=1}^K \sum_{\ell=1}^L  \nabla_{\boldsymbol{\theta}} \underline{R}_{k,\ell}(\boldsymbol{\theta}).
\end{equation}

Furthermore, the gradient of $\mathcal{G}(\boldsymbol{\theta})$ with respect to $\boldsymbol{\theta}$ is obtained as
\begin{equation}
\begin{split}
   & \nabla_{\boldsymbol{\theta}} \mathcal{G}({\boldsymbol{\theta}})   = 2 \rho K  \Big( \sum_{k=1}^K   \frac{1}{w_k^2} \sum_{\ell=1}^L \underline{R}_{k,\ell} ({\boldsymbol{\theta}}) \sum_{m=1}^L \nabla_{\boldsymbol{\theta}} \underline{R}_{k,m} ({\boldsymbol{\theta}} ) \Big)  \\
     & - 2 \Big( \sum_{k=1}^K  \frac{1}{w_k}  \sum_{\ell=1}^L \underline{R}_{k,\ell} ({\boldsymbol{\theta}}) \Big) \Big(  \sum_{i=1}^K  \frac{1}{w_i}  \sum_{m=1}^L \nabla_{\boldsymbol{\theta}} \underline{R}_{i,m} ({\boldsymbol{\theta}} ) \Big).
    \end{split}
\end{equation}
Therefore, the gradient $\mathcal{H}_{\gamma, \omega}$ with respect to $\boldsymbol{\theta}$ is
\begin{equation}
    \nabla_{\boldsymbol{\theta}} \mathcal{H}_{\gamma, \omega} (\boldsymbol{\theta})  = \frac{\nabla_{\boldsymbol{\theta}} \underline{R}_\text{sum} (\boldsymbol{\theta})}{P_\text{tot}} -  \gamma \nabla_{\boldsymbol{\theta}} \mathcal{G}(\boldsymbol{\theta}) 
     - \frac{1}{ \omega}  \mathcal{G}(\boldsymbol{\theta}) \nabla_{\boldsymbol{\theta}} \mathcal{G}(\boldsymbol{\theta}).
\end{equation}

Therefore, the update of $\boldsymbol{\theta}$ should follow
\begin{equation}
   \bar{\boldsymbol{\theta}}^{(t+1)} =\boldsymbol{\theta}^{(t)} + \alpha_{\boldsymbol{\theta}}^{(t)} \nabla_{\boldsymbol{\theta}} \mathcal{H}_{\gamma, \omega} (\boldsymbol{\theta}^{(t)}),
   \label{eq:F_RIS_update_EE}
\end{equation}
and the projection onto the constraint~\eqref{ps_con_EEf} can be performed as
\begin{equation}
    \boldsymbol{\theta}^{(t+1)} =  \text{diag}(|\bar{\boldsymbol{\theta}}^{(t+1)}|)^{-1} \bar{\boldsymbol{\theta}}^{(t+1)}.
    \label{eq:proj_theta}
\end{equation}
\vspace{-0.3cm}
\subsection{Receive Combiners}
Assuming that the values of $\mathbf{D},\mathbf{A}$ and $\boldsymbol{\Theta}$ are fixed, the combiner optimization sub-problem can be written as
\begin{subequations}
\begin{align}
            \max_{  \tilde{\mathbf{C}}}  \ \ \  & \underline{\eta} (\tilde{\mathbf{C}}), \\
        \text{s.t.} \ \ \ & \eqref{F_const}, \  \eqref{lb_constraint},
\end{align} 
\label{eq:optm_EE_W_RF} 
\end{subequations} 
with the following augmented Lagrangian function:
\begin{equation}
       \mathcal{H}_{\gamma, \omega} ( \tilde{\mathbf{C}})  = \frac{1}{P_\text{tot}} \underline{R}_\text{sum} ( \tilde{\mathbf{C}}) - \gamma \mathcal{G}( \tilde{\mathbf{C}}) 
     - \frac{1}{2 \omega} \ \mathcal{G}^2( \tilde{\mathbf{C}}) .
     \label{eq:aug_Lag_C}
\end{equation}

We start by expressing $\underline{R}_{k,\ell} (\tilde{\mathbf{c}}_{k,\ell})$ as
\begin{equation}
 \begin{split}
    & \underline{R}_{k,\ell} (\tilde{\mathbf{c}}_{k,\ell})  =  \log_2 \bigg( \varrho_k |(\mathbf{J}_{k}^{k,\ell})^\mathsf{H} \tilde{\mathbf{c}}_{k,\ell}|^2  \\
    & + \underset{[i,v] \neq [k,\ell]}{\sum_{i=1}^K \sum_{ {v=1}}^L}  \big(|(\mathbf{J}_{k}^{i,v})^\mathsf{H} \tilde{\mathbf{c}}_{k,\ell}| +  \delta_k \sqrt{N_\text{RIS}} || \mathbf{d}_{i,v}||_2  \big)^2 + \sigma^2 \bigg)  \\
     & -   \log_2 \bigg( \underset{[i,v] \neq [k,\ell]}{\sum_{i=1}^K \sum_{ {v=1}}^L}  \big(|(\mathbf{J}_{k}^{i,v})^\mathsf{H} \tilde{\mathbf{c}}_{k,\ell}| +  \delta_k \sqrt{N_\text{RIS}} || \mathbf{d}_{i,v}||_2 \big)^2 \bigg),
    \end{split}
\end{equation}
 with  $\mathbf{J}_{k}^{i,v} =\hat{\mathbf{G}}_{k} (\boldsymbol{\Theta}) \mathbf{A} \mathbf{d}_{i,v }$.

The gradient of $ \underline{R}_{k,\ell} (\tilde{\mathbf{c}}_{k,\ell})$ with respect to $\tilde{\mathbf{c}}_{k,\ell}$can then be obtained as expressed in~\eqref{eq:grad_r_w} on the next page,
\begin{figure*} [b]
  \hrulefill
 \begin{equation}
\begin{split}
    \nabla_{\tilde{\mathbf{c}}_{k,\ell}} \underline{R}_{k,\ell} (\tilde{\mathbf{c}}_{k,\ell}) = \frac{2}{\ln 2} \bigg[ & \frac{\varrho_k \mathbf{J}_{k}^{i,v} (\mathbf{J}_{k}^{i,v})^\mathsf{H} +\underset{[i,v] \neq [k,\ell]}{\sum_{i=1}^K \sum_{ {v=1}}^L} \big(1 + \frac{  \delta_k \sqrt{N_\text{RIS}}|| \mathbf{d}_{i,v}||_2}{|(\mathbf{J}_{k}^{i,v})^\mathsf{H} \tilde{\mathbf{c}}_{k,\ell}   } \big) \mathbf{J}_{k}^{i,v} (\mathbf{J}_{k}^{i,v})^\mathsf{H} }{\varrho_k |(\mathbf{J}_{k}^{k,\ell})^\mathsf{H} \tilde{\mathbf{c}}_{k,\ell}|^2  + \underset{[i,v] \neq [k,\ell]}{\sum_{i=1}^K \sum_{ {v=1}}^L}  \big(|(\mathbf{J}_{k}^{i,v})^\mathsf{H} \tilde{\mathbf{c}}_{k,\ell}| +  \delta_k \sqrt{N_\text{RIS}} || \mathbf{d}_{i,v}||_2  \big)^2 + \sigma^2} \\
    &- \frac{\underset{[i,v] \neq [k,\ell]}{\sum_{i=1}^K \sum_{ {v=1}}^L} \big(1 + \frac{  \delta_k \sqrt{N_\text{RIS}}|| \mathbf{d}_{i,v}||_2}{|(\mathbf{J}_{k}^{i,v})^\mathsf{H} \tilde{\mathbf{c}}_{k,\ell}   } \big) \mathbf{J}_{k}^{i,v} (\mathbf{J}_{k}^{i,v})^\mathsf{H} }{ \underset{[i,v] \neq [k,\ell]}{\sum_{i=1}^K \sum_{ {v=1}}^L}  \big(|(\mathbf{J}_{k}^{i,v})^\mathsf{H} \tilde{\mathbf{c}}_{k,\ell}| +  \delta_k \sqrt{N_\text{RIS}} || \mathbf{d}_{i,v}||_2  \big)^2 + \sigma^2}  \bigg] \tilde{\mathbf{c}}_{k,\ell},
    \label{eq:grad_r_w}
    \end{split}
\end{equation}
\end{figure*}
and that of $ \underline{R}_{\text{sum}} (\tilde{\mathbf{C}})$ with respect to $\tilde{\mathbf{C}}$ is
\begin{equation}
    \nabla_{\tilde{\mathbf{C}}} \underline{R}_{\text{sum}} (\tilde{\mathbf{C}}) = [ \nabla_{\tilde{\mathbf{c}}_{k,\ell}} \underline{R}_{1,1} (\tilde{\mathbf{c}}_{1,1}),\dots, \nabla_{\tilde{\mathbf{c}}_{k,\ell}} \underline{R}_{K,L} (\tilde{\mathbf{c}}_{K,L})].
    \vspace{-0.5cm}
\end{equation}

Additionally, the gradient of $\mathcal{G}(\tilde{\mathbf{C}})$ with respect to $\tilde{\mathbf{C}}$ is found as
\begin{equation}
\begin{split}
    & \nabla_{\tilde{\mathbf{C}}} \mathcal{G}({\tilde{\mathbf{C}}})   = 2 \rho K  \Big( \sum_{k=1}^K   \frac{1}{w_k^2} \sum_{\ell=1}^L \underline{R}_{k,\ell} ({\tilde{\mathbf{c}}}_{k,\ell} ) \\
    & \times \sum_{m=1}^L \nabla_\mathbf{D} \underline{R}_{k,m} ({\tilde{\mathbf{c}}}_{k,m} ) \Big)  \\
     & - 2 \Big( \sum_{k=1}^K  \frac{1}{w_k}  \sum_{\ell=1}^L \underline{R}_{k,\ell} ({\tilde{\mathbf{c}}}_{k,\ell} ) \Big) \Big(  \sum_{i=1}^K  \frac{1}{w_i}  \sum_{m=1}^L \nabla_\mathbf{D} \underline{R}_{i,m} ({\tilde{\mathbf{c}}}_{i,m}) \Big).
    \end{split}
\end{equation}
Thus, the gradient of $\mathcal{H}_{\gamma, \omega}  ( \tilde{\mathbf{C}})$ with respect to $\tilde{\mathbf{C}}$ is
\begin{equation}
       \nabla_{\tilde{\mathbf{C}}} \mathcal{H}_{\gamma, \omega} ( \tilde{\mathbf{C}})   = \frac{\nabla_{\tilde{\mathbf{C}}} \underline{R}_\text{sum} (\tilde{\mathbf{C}})}{P_\text{tot}} - \gamma \nabla_{\tilde{\mathbf{C}}} \mathcal{G}(\tilde{\mathbf{C}}) 
     - \frac{1}{ \omega}  \mathcal{G}(\tilde{\mathbf{C}}) \nabla_{\tilde{\mathbf{C}}}\mathcal{G}(\tilde{\mathbf{C}}).
\end{equation}
Thus, the update of $\tilde{\mathbf{C}}$ should follow
\begin{equation}
    \bar{\mathbf{C}}^{(t+1)} =   \tilde{\mathbf{C}}^{(t)} + \alpha_{\mathbf{C}}^{(t)} \nabla_{\tilde{\mathbf{C}}} \mathcal{H}_{\gamma, \omega} ( \tilde{\mathbf{C}}^{(t)}),
   \label{eq:W_RF_update_EE}
\end{equation}
and the projection onto the constraint~\eqref{lb_constraint} can be done by normalizing each column of $\tilde{\mathbf{C}}$ to have unit norm, i.e.,
\begin{equation}
\tilde{\mathbf{c}}_{k,\ell}^{(t+1)}   = \bar{\mathbf{c}}_{k,\ell}^{(t+1)} / || \bar{\mathbf{c}}_{k,\ell}^{(t+1)} ||_2,  \ \forall \ell, k.
    \label{eq:W_RF_proj_EE}
\end{equation}

\vspace{-0.3cm}

\subsection{The Slack Variable $\mu$ } 
The optimization problem with respect to the slack variable $\mu$ is
\begin{subequations}
\begin{align}
            \max_{ \mu \geq 0}  \ \ \ & \mathcal{H}_{\gamma, \omega}
     (\mu).
\end{align} 
\end{subequations}
 
 To find the optimal value of $\mu$, we find the partial derivative $ \frac{\partial \mathcal{H}_{\gamma, \omega} (\mu) }{\partial \mu} $ and equate it to zero to obtain
 \begin{equation}
     \frac{\partial \mathcal{H}_{\gamma, \omega} (\mu)}{\partial \mu} = - \gamma - \frac{1}{\omega}  \mathcal{G} (\mu) = 0,
 \end{equation}
which has a solution of
\begin{equation}
\begin{split}
    \bar{\mu} & =  \Big(  \sum_{k=1}^K  \frac{1}{w_k} \sum_{\ell=1}^L \underline{R}_{k,\ell} \Big)^2  - \rho K  \sum_{k=1}^K  \Big( \frac{1}{w_k} \sum_{\ell=1}^L \underline{R}_{k,\ell} \Big)^2  
    - \gamma \omega.
    \end{split}
\end{equation}
Then the update of $\mu$ after projecting onto the feasible set is
\begin{equation}
    \mu^{(t+1)} = \max \big( 0,\bar{\mu} \big).
    \label{eq:mu_update}
\end{equation}

Algorithm \ref{alg:S2} summarizes the proposed penalty-based projected gradient ascent approach.

\begin{algorithm}  
\caption{AO algorithm based on projected gradient ascent and PDD method}
\label{alg:S2}                          
\begin{algorithmic} [1]                   
\State \textbf{Input:}  
$\rho \in [0,1]$,
$\gamma \geq 0$,
$\omega > 0$,
$\psi \in [0,1]$,
$\mathbf{D}^{(0)} $, 
$\mathbf{A}^{(0)}$, 
$\boldsymbol{\Theta}^{(0)} $, 
$\tilde{\mathbf{C}}^{(0)}$, 
$\mu^{(0)}=0$,
$\underline{\eta}^*$,
$\rho \in [0,1]$,
$\gamma$,
$\omega$,
$t = 0$,
 and $\varepsilon > 0$.
\State \textbf{While} $\mu^{(t)} > 0$  \textbf{do} \{
\State \indent \textbf{do \{ } 
\State \indent \indent Update $\bar{\mathbf{D}}^{(t+1)} \leftarrow \mathbf{D}^{(t)} + \alpha_{\mathbf{D}}^{(t)}  \nabla_{\mathbf{D}} \mathcal{H}_{\gamma, \omega} (\mathbf{D}^{(t)})$ \indent \indent using~\eqref{eq:F_BB_update_f}. 
\State \indent \indent Project onto the constraint \eqref{Pow_con} using \eqref{eq:proj_f_bb}.

\Statex 
\State \indent \indent Update $\bar{\mathbf{a}}^{(t+1)} \leftarrow \mathbf{a}^{(t)} + \alpha_{\mathbf{a}}^{(t)}  \nabla_\mathbf{a} \mathcal{H}_{\gamma, \omega}   (\mathbf{a}^{(t)})$ \indent 
 \indent using~\eqref{eq:F_RF_update_EE}. 
\State \indent \indent Project onto the constraint \eqref{CM_BS} using \eqref{eq:proj_f_RF}.

\Statex  
\State \indent \indent Update $\bar{\boldsymbol{\theta}}^{(t+1)} \leftarrow \boldsymbol{\theta}^{(t)} + \alpha_{\boldsymbol{\theta}}^{(t)} \nabla_{\boldsymbol{\theta}} \mathcal{H}_{\gamma, \omega}(\boldsymbol{\theta}^{(t)})$ \indent \indent using~\eqref{eq:F_RIS_update_EE}.
\State \indent \indent  Project onto the constraint \eqref{ps_con_EEf} using \eqref{eq:proj_theta}.

\Statex  
\State \indent \indent Update $\bar{\mathbf{C}}^{(t+1)} \leftarrow   \tilde{\mathbf{C}}^{(t)} + \alpha_{\mathbf{C}}^{(t)} \nabla \mathcal{H}_{\gamma, \omega} (\tilde{\mathbf{C}}^{(t)})$ \indent \indent using~\eqref{eq:W_RF_update_EE}.
\State \indent \indent  Project onto the constraint \eqref{lb_constraint} using \eqref{eq:W_RF_proj_EE}.

\Statex
\State \indent \indent Update $\mu^{(t+1)}$ using \eqref{eq:mu_update}. 
\Statex
\State \indent \indent Update $t \leftarrow t + 1.$
\State \indent  \textbf{ \} } 
\State \indent  \textbf{until} \big| $\mathcal{H}_{\gamma, \omega} (\mathbf{D}^{(t)}, 
\mathbf{A}^{(t)}, 
\boldsymbol{\Theta}^{(t)} , 
\tilde{\mathbf{C}}^{(t)}, \mu^{(t)}) -\mathcal{H}_{\gamma, \omega}  ( $ \indent $\mathbf{D}^{(t-1)},  
 \mathbf{A}^{(t-1)}, 
\boldsymbol{\Theta}^{(t-1)} , 
\tilde{\mathbf{C}}^{(t-1)},\mu^{(t-1)}) \big| \leq \varepsilon$.
\State \indent Update $\gamma \leftarrow \gamma + \frac{1}{\omega} \mathcal{H}_{\gamma, \omega}
    (\mathbf{D}^{(t)},  
\mathbf{A}^{(t)},  
\boldsymbol{\Theta}^{(t)} , 
\tilde{\mathbf{C}}^{(t)}, $ \indent $ \mu^{(t)}) $.
\State \indent Update $ \omega \leftarrow \psi \omega$.
\State \}
\State \textbf{Output:} $\underline{\eta}^* = \underline{\eta} (\mathbf{D}^{(t)},\mathbf{A}^{(t)}, \boldsymbol{\Theta}^{(t)}, \tilde{\mathbf{C}}^{(t)} )$.
\end{algorithmic}
\end{algorithm}
\vspace{-0.3cm}
\subsection{Convergence} 
The employed algorithms (i.e., AO, projected gradient ascent, and the PDD method) are proven to converge to a locally optimal solution, as was demonstrated in~\cite{AO_con_proof, PGA_con_proof, pen}. Due to the nature of the targeted optimization problem, finding a global solution poses a challenge. However, a satisfactory locally optimal solution can be obtained by solving the optimization problem multiple times with different random initializations, as will be illustrated in Section~\ref{ssec:conv}.

\subsection{Complexity Analysis} \label{sec:complexity}
In this subsection, we provide an analysis of the computational complexity of the proposed method, defined as the number of required complex-valued multiplications. Each variable update and projection for Algorithm~\ref{alg:S2} is dominated by specific vector multiplications, including inner and outer products. Table~\ref{tab:complexity} outlines the complexity of updating and projecting each variable.

\begin{table}
\footnotesize
\centering
\caption{Computational complexity.}
\begin{tabular} {|| m{1cm} | m{4.5cm} | m{1.7cm} ||} 
 \hline  Variable & Update complexity & Projection complexity  \\  [0.5ex] 
 \hline 
   $\mathbf{D}$ & $\mathcal{O}\big( KLM(MLN_\mathsf{T}N_\text{RIS} + N_\text{RIS}^2 + N_\text{RIS} N_\mathsf{R} + KL) \big)$ & $\mathcal{O}\big( K^2 L^2 M \big)$\\
 \hline
  $\mathbf{A}$ & $\mathcal{O}\big( K^2 L^2MN_\mathsf{T} (MN_\mathsf{T}N_\text{RIS} + N_\text{RIS}^2 + N_\text{RIS} N_\mathsf{R}) \big)$ & $\mathcal{O}\big( MN_\mathsf{T} \big)$\\
 \hline
 $\boldsymbol{\Theta}$ & $\mathcal{O}\big( K^2 L^2 N_\text{RIS} (N_\mathsf{R}+ M^2 N_\mathsf{T} +  N_\text{RIS} \big)$ & $\mathcal{O}\big( N_\text{RIS} \big)$\\
 \hline
   $\tilde{\mathbf{C}}$ & $\mathcal{O}\big( K^2 L^2 N_\mathsf{R} ( N_\text{RIS}^2 + M^2 N_\mathsf{T} + N_\mathsf{R}) \big)$ & $\mathcal{O}\big( KLN_\mathsf{R} \big)$\\
 \hline
\end{tabular}
\label{tab:complexity}
\vspace{-0.2cm}
\end{table}

In practical systems, the number of RIS elements is significantly larger than the number the number of BS antennas, the number of UE antennas, and the number of data streams i.e., $N_\text{RIS} \gg \max(MN_\mathsf{T},N_\mathsf{R},LK)$. In this scenario, the computational complexity of one inner iteration of Algorithm~\ref{alg:S2} is primarily dominated by the term $\mathcal{O}\big( L^2 K^2N_\text{RIS}^2 (MN_\mathsf{T} + N_\mathsf{R}) \big)$.
\section{Numerical Simulations} \label{sec:sim}
In this section, we conduct extensive simulations to demonstrate the performance achievable via the proposed method. We consider two state-of-the-art benchmark schemes for comparison, the EE-only maximization scheme of~\cite{EE_opt3} (both without and with rate constraints), and the EE fairness maximization scheme of~\cite{fairness_3}.

In our simulations, we average the performance metric over 1000 independent channel realizations, with small-scale fading, user locations, and QoS weights randomly varied within confined ranges. Table~\ref{tab:sim_params} provides a summary of the simulation parameters; these are adopted in all simulations unless stated otherwise.
\begin{table}
\footnotesize
\centering
\caption{Simulation parameters.}
\begin{tabular} {||m{1.1cm} |m{1.0cm} || m{1.9cm} | m{2.2cm} ||} 
 \hline 
 Parameter & Value & Parameter & Value \\  [0.5ex] 
 \hline 
  $M$ & 8 &  $P_{\theta}$  &  \unit[1]{dBm}  \\
 \hline
  $N_\mathsf{T}$ & 4 & $\xi$ & 1.2 \\
 \hline
  $N_\text{RIS}$ &  64 &  $\rho$ & 0.75 \\
    \hline
    $K$ & 4 & $\varepsilon$ & $1\times 10^{-3}$ \\
    \hline
    $L$ & 2 &  $\sigma^2$ &  \unit[-37]{dBm} \\
    \hline
      $N_\mathsf{R}$ &  4 &  $w_k$ range & $[1,5]$ \\
    \hline
    Carrier frequency & $ \unit[28]{GHz}$ & Location of the reference AE at the BS & $\unit[(0,0,10)]{m}$ \\
    \hline
    Total bandwidth & $ \unit[200]{MHz}$ & Location of the center of the RIS & $\unit[(15,-15,5)]{m}$ \\
    \hline
    $P_\text{max}$ & \unit[40]{dBm}&  $k$-th UE’s location range & $[\unit[(15,-15,0)]{m},$ $\unit[(30,15,2)]{m}]$ \\
    \hline
    $P_\text{BS}$ & \unit[9]{dBW} & $k$-th user’s imperfect CSI bound & $0.2 || \hat{\mathbf{H}}_k||_\mathsf{F}$\\
    \hline
    $P_{\text{UE},k}$ & \unit[5]{dBm} & Rate constraints for~\cite{EE_opt3} & $ R_k \geq  0.1 w_k \log_2 \big( \frac{P_\text{max}}{K\sigma^2} \big)$ \\
    \hline
     $P_{\text{RF,T}}$  & \unit[1]{dBm} &  & \\
     \hline
\end{tabular}
\label{tab:sim_params}
\vspace{-0.5cm}
\end{table}

\subsection{Convergence} \label{ssec:conv}
\begin{figure*}
  \centering
  \begin{tabular}{c c}
    \includegraphics[width=0.83\columnwidth]{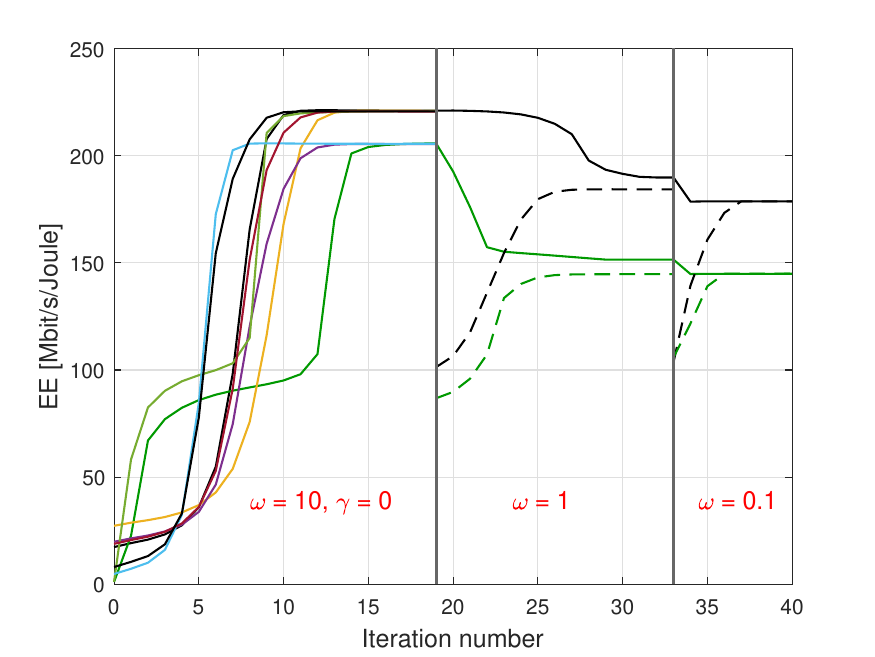} &
      \includegraphics[width=0.83\columnwidth]{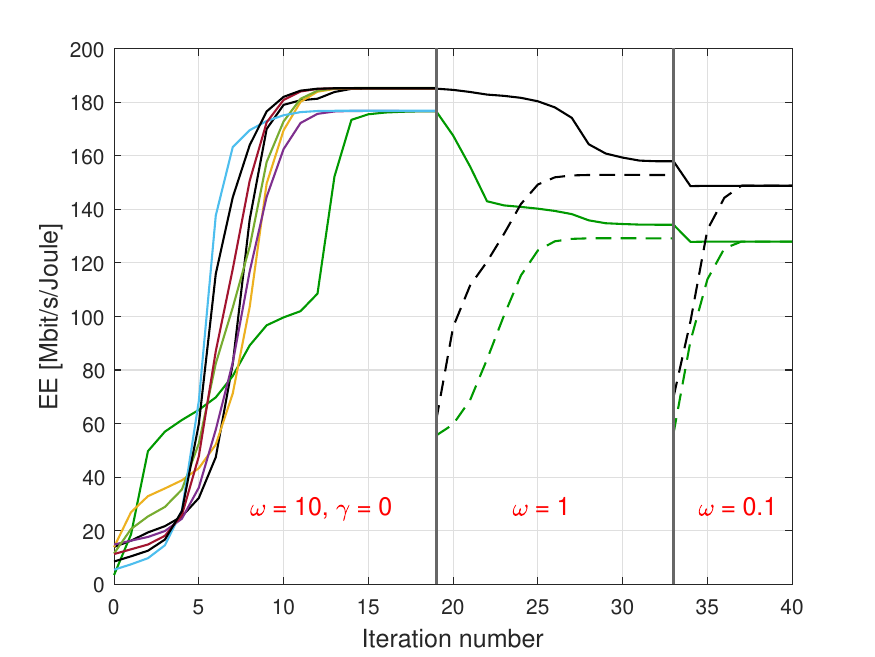}  \\
      \scriptsize (a) Perfect CSI.   &
      \scriptsize (b) Imperfect CSI.  \\
  \end{tabular}
    \medskip
  \caption{EE (solid) and the augmented Lagrangian function (dashed) versus iteration number for the cases of perfect and imperfect CSI.}
  \label{fig:conv}
  \vspace{-0.5cm}
\end{figure*}

We first examine the convergence of the algorithm by plotting the EE and the augmented Lagrangian function versus the iteration number for both perfect and imperfect CSI cases in Fig.~\ref{fig:conv}. We use an initial step size of one to update each variable and progressively decrease it by half in the case of overshooting (i.e., when we observe a decrease in the objective function). We initiate the process by setting the values of $\gamma$ and $\omega$ to 0 and 10, respectively. Using eight random initial points that satisfy the constraint set, we observe convergence towards two sub-optimal points due to the non-convexity of the problem. In this initial stage, the weight of the constraint in the augmented Lagrangian function is small. Therefore, the algorithm primarily focuses on optimizing the EE part of the objective function. Since the value of the EE in this case is almost equal to that of the augmented Lagrangian function, we omit the plot of the augmented Lagrangian function. Subsequently, the parameter $\omega$ is reduced by a factor of 10 to accentuate the constraint more in the augmented Lagrangian function until it is satisfied.

At the final iteration, the augmented Lagrangian function converges to the value of the EE objective function. This convergence is expected, as both functions should have the same value when the constraint~\eqref{F_const} is satisfied. The algorithm eventually converges to two EE values (179 and \unit[145]{Mbit/sec/Joule} for the perfect CSI scenario, and 149 and \unit[128]{Mbit/sec/Joule} for the imperfect CSI scenario). To ensure convergence to a satisfactory solution, we utilize the random initialization technique in all simulations by employing five initial points and selecting the highest final EE value.

\subsection{Effect of Varying the Fairness Requirement} \label{sec:CEE}

\begin{figure} 
         \centering
         \includegraphics[width=0.83\columnwidth]{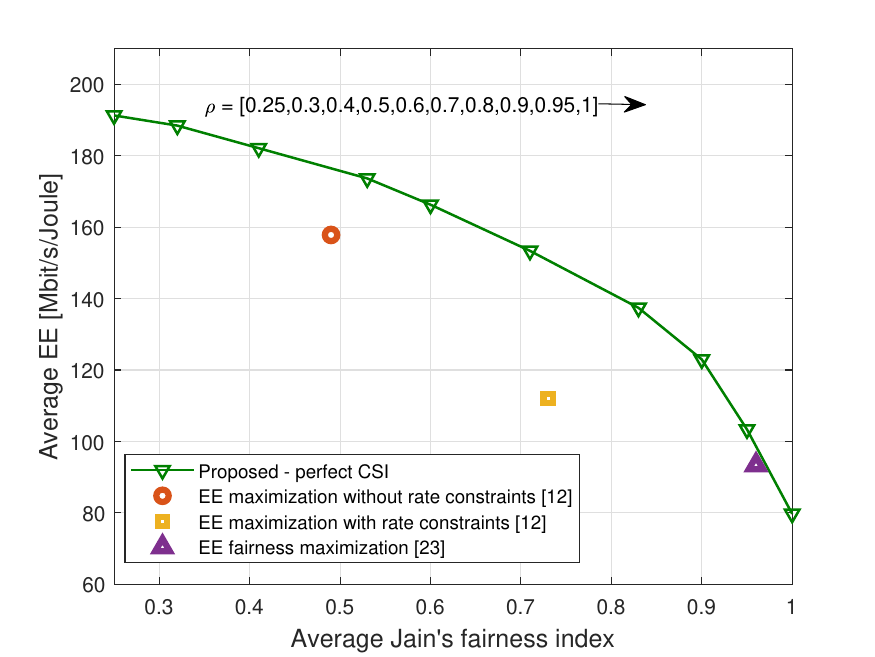}
        \caption{Average EE versus Jain's fairness index for different values of $\rho$.}
        \label{fig:rho}
        \vspace{-0.4cm}
\end{figure}

Fig.~\ref{fig:rho} demonstrates the EE-fairness trade-off achieved by varying the value of $\rho$. The general trend indicates that as $\rho$ increases, the EE decreases, and the Jain's fairness index increases. This relationship arises because $\rho$ controls the fairness requirement of the system at the cost of reducing the EE value. At one extreme, perfect fairness (Jain's fairness index of one) can be achieved when $\rho = 1$, but this comes at the expense of reducing the EE by 63\% of its optimal value when no fairness is imposed. At the other extreme, the maximum EE can be achieved with very poor fairness. The value of $\rho$ can be tuned to achieve a favorable and flexible trade-off between these two extremes.

Compared to the existing EE maximization method of~\cite{EE_opt3}, the proposed algorithm can achieve a better EE value with the same fairness index. This improvement is attributed to our proposed system model and transmission scheme, which are more suitable for EE applications than that of~\cite{EE_opt3}. While the poor fairness in~\cite{EE_opt3} can be enhanced by incorporating QoS constraints, it comes with a significant reduction in the EE. On the other hand, the fairness algorithm proposed in~\cite{fairness_3} can achieve a point close to the Pareto-optimal front generated by varying $\rho$ in our proposed algorithm. Nevertheless, the proposed method offers the flexibility of tuning $\rho$ to achieve the desired operating point of EE and fairness based on available resources and specified QoS requirements. 
\vspace{-0.3cm}
\subsection{Effect of Varying the CSI Error Bound} 
\begin{figure} 
         \centering
         \includegraphics[width=0.83\columnwidth]{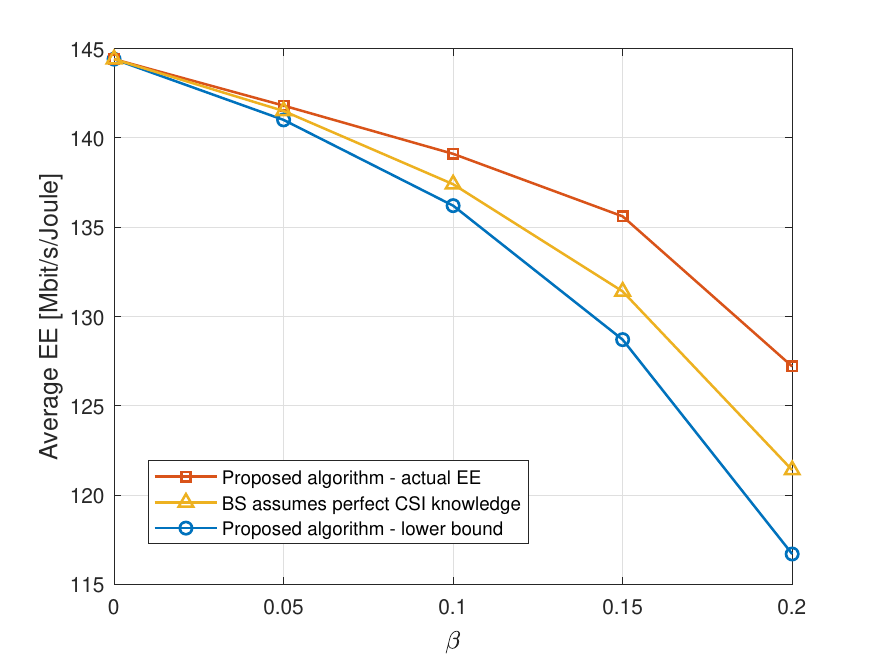}
        \caption{Behavior of the average EE with increasing CSI error bound.}
        \label{fig:CSI_bound}
        \vspace{-0.7cm}
\end{figure}

Fig.~\ref{fig:CSI_bound} shows the effect of increasing the CSI error bound, where the CSI error bound of each user is modeled as 
\begin{equation}
    \delta_k = \beta || \hat{\mathbf{H}}_k ||_\mathsf{F},
    \label{eq:last_eq}
\end{equation}
for some $\beta \in [0,0.2]$. We plot three curves: the true EE given in~\eqref{eq:EE} by optimizing the lower bound in~\eqref{eq:EE_lb}, the lower bound on the EE given in~\eqref{eq:EE_lb}, and the EE in the case when the BS assumes no error exists in the CSI.

It can be noticed that, compared to the case where the BS assumes perfect CSI knowledge, optimizing the lower bound instead of the original function can offer improvements in EE. However, this improvement comes at an additional cost in complexity, as more terms must be included in the gradient calculations. In conclusion, if the imperfect CSI bound is small, it might be more efficient for the BS to assume perfect CSI knowledge, as the improvement is not significant. Conversely, optimizing the EE lower bound offers a notable improvement for scenarios where the CSI uncertainty is relatively high.

\subsection{Effect of Varying the Transmit Power Budget}

\begin{figure*}
  \centering
  \begin{tabular}{c c}
    \includegraphics[width=0.83\columnwidth]{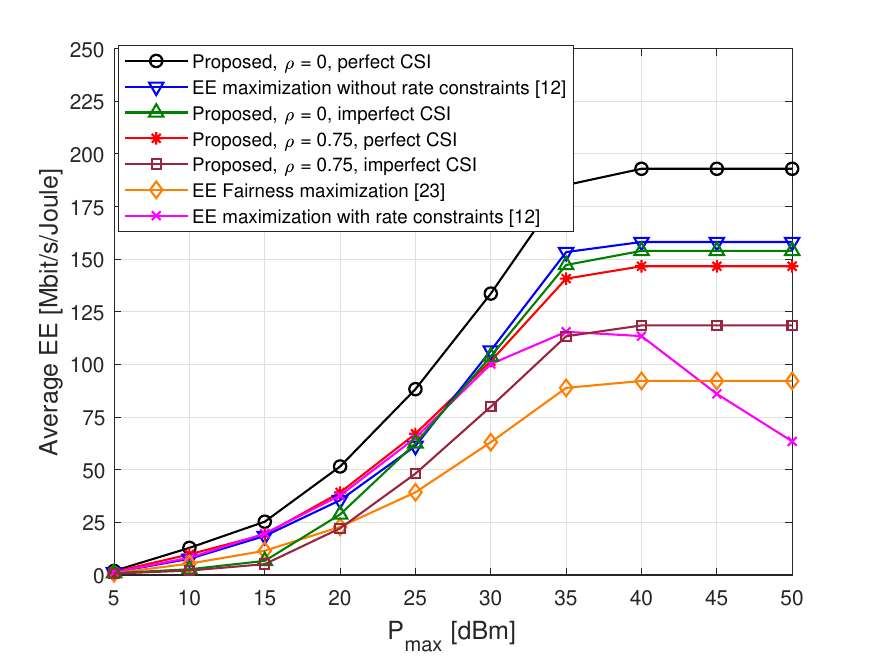} &
      \includegraphics[width=0.83\columnwidth]{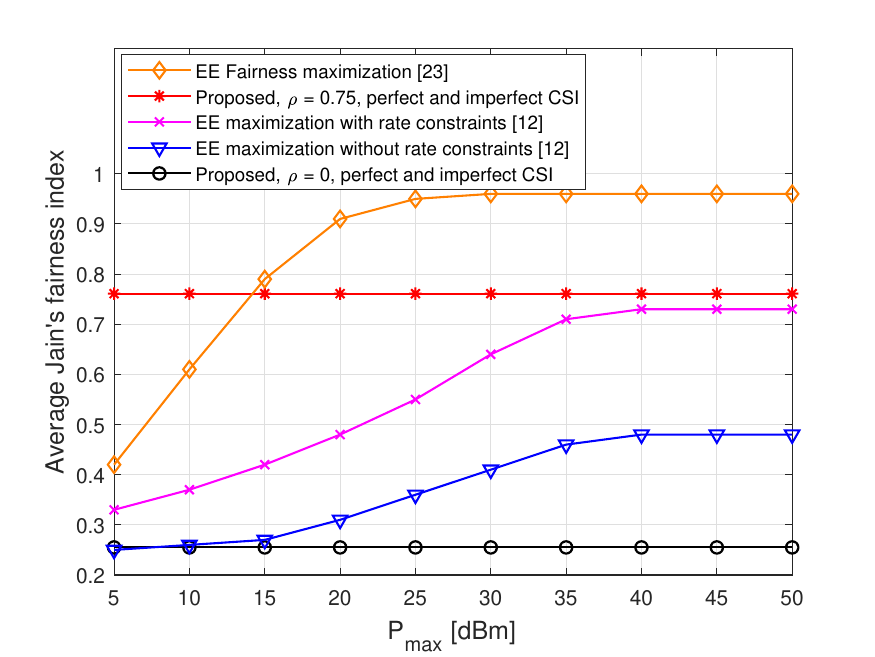}  \\
      \scriptsize (a) Average EE.   &
      \scriptsize (b) Average Jain's fairness index.  \\
  \end{tabular}
    \medskip
  \caption{Average EE and Jain's fairness index as the transmit power budget increases.}
  \label{fig:P_max}
  \vspace{-0.7cm}
\end{figure*}

Fig.~\ref{fig:P_max} shows the average EE and Jain's fairness index as the maximum transmit power budget increases. It can be noted that both EE and fairness saturate when the transmit power budget exceeds \unit[40]{dBm}. This indicates that the extra power is not being used as it does not help in optimizing the EE or fairness. This can be attributed to the fact that the numerator of the EE objective function is a logarithmic function of the transmit power while the denominator is a linear function of the transmit power. In other terms, the rate of increase of the numerator (sum-rate) decreases with increasing $P_\text{max}$, while that of the denominator (total power consumption) is increasing in a constant rate.

Furthermore, it can be noted that $\rho=0$ in the proposed method can achieve a better EE than the state-of-art methods for three main reasons. First, the AoSA architecture is more energy-efficient than the fully-connected one. This conclusion matches observation made in~\cite{AoSA_motivation}, where the AoSA and fully-connected architectures are compared from this perspective. Second, the proposed system has a greater capability of increasing the EE as it features two additional sets of variables: the analog precoder at the BS and the combiners at the UEs. Third, unlike the benchmark methods, ZF and orthogonal transmission have not been assumed. While orthogonal transmission clearly limits the system capabilities, ZF is not ideal in the case of correlated, ill-conditioned, and rank-deficient mmWave channels. Nevertheless, considering only EE leads to very poor fairness, as in most practical cases all resources are allocated to serve one user only. This issue can be resolved by targeting fairness optimization directly as in~\cite{fairness_3}, however, the resulting EE is poorer than that achieved by other methods. The proposed approach can achieve a better EE and fairness trade-off than these existing approaches, as optimizing the EE with a Jain's fairness index requirement of 0.75 can lead to highly favorable values for both the EE and user fairness. On the one hand, although incorporating rate constraints can help to improve fairness, it can degrade the system EE significantly when these constraints depend on the transmit power budget. Nevertheless, setting the values for the rate constraints independently of the transmit power budget may lead to feasibility issues for the optimization problem.

\subsection{Effect of Varying the Number of RIS elements}
\begin{figure} 
         \centering
         \includegraphics[width=0.83\columnwidth]{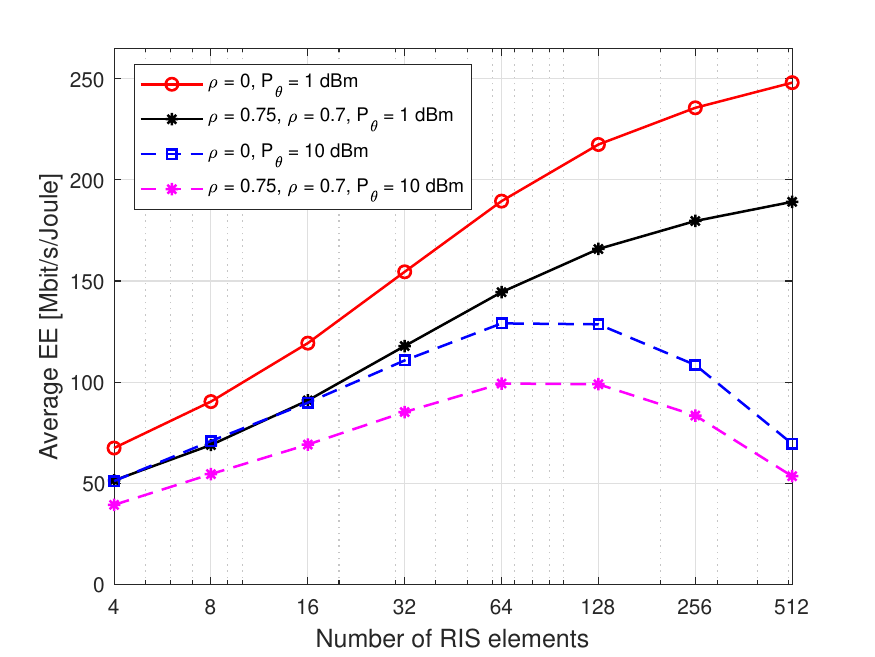}
        \caption{Average EE as the number of RIS elements increases.}
        \label{fig:N_RIS}
        \vspace{-0.7cm}
\end{figure}

In Fig. \ref{fig:N_RIS}, we plot the average EE as the number of RIS elements increases for two values of the power consumption at each RIS element: $P_\theta = \unit[1]{dBm}$ and  $P_\theta = \unit[10]{dBm}$. We can notice that adding extra RIS elements provides a substantial increase in the EE for all cases when the number of RIS elements is small. However, as the number of RIS elements becomes larger than 64, this increase begins to slow down in the case of $P_\theta = \unit[1]{dBm}$, while the average EE starts to decrease for $P_\theta = \unit[10]{dBm}$. This eventual decrease in the EE occurs due to the fact that although more RIS elements offer more degrees of freedom to enhance the system performance, it provides a logarithmic increase in the numerator of the EE but a linear increase in the denominator.

\section{Conclusion} \label{sec:conc}
This paper has presented a method for robust design of the hybrid analog and digital precoder at the BS, RIS reflection matrix, and digital combiners at the UEs in order to maximize the system EE while imposinga specified Jain's fairness index in RIS-assisted mmWave systems with imperfect channel state information. To achieve this, a lower bound based on the triangle and Cauchy-Schwarz inequalities was employed. The penalty dual decomposition method was then used to address the challenging fairness constraint, and the projected gradient ascent method was utilized to obtain the solution to the optimization problem. The simulation results show that the proposed method can provide a very good tradeoff between EE and fairness, while also offering a high degree of flexibility in tuning the EE and user fairness in order to prioritize one of these metrics over the other.

 \appendices
\vspace{-0.3cm}
  \section{Proof of Proposition \ref{prop:1}} \label{AppndA}
  The signal-to-interference-plus-noise ratio (SINR) in~\eqref{eq:Rk} consists of three components, the desired signal term, the interference term, and the noise term.
  
  To obtain a lower bound on the desired signal term, we seek the error matrix that minimizes that term by solving the following convex optimization problem:
\begin{subequations}
\begin{align}
  \min_ {\boldsymbol{\Lambda}_k} \ \ & \mathcal{D}_k(\boldsymbol{\Lambda}_k) = \big| \mathbf{c}_{k,\ell}^\mathsf{H} \big( \hat{\mathbf{G}}_{k}  (\boldsymbol{\Theta}) + \boldsymbol{\Lambda}_k \big) \mathbf{A} \mathbf{d}_{k,\ell } \big|^2, \label{eq:apx_of} \\
   \text{s.t.} \ \ & || \boldsymbol{\Delta}_k ||_\text{F} \leq \delta_k.
  \label{eq:op_app}
   \end{align}
  \end{subequations}
  We can rewrite the objective function~\eqref{eq:apx_of} as
  \begin{equation}
  \begin{split}
       \mathcal{D}_k(\boldsymbol{\Delta}_k)  = \big| \big( (\mathbf{A} \mathbf{d}_{k,\ell }  )^\mathsf{T} \otimes \mathbf{c}_{k,\ell}^\mathsf{H} \big) \big( \hat{\mathbf{H}}_{k} + \boldsymbol{\Delta}_k \big)   \boldsymbol{\theta} \big|^2.
         \end{split}
      \label{eq:obj_fun_proof}
  \end{equation}
  We observe that~\eqref{eq:obj_fun_proof} attains its minimum when $\boldsymbol{\Delta}_k = - \hat{\mathbf{H}}_{k}$. By projecting this solution onto the feasible set, the optimal solution is found as $\boldsymbol{\Delta}_k^* = - \min \big( 1, \delta_k / || \hat{\mathbf{H}}_{k}  ||_{\mathsf{F}} \big) \hat{\mathbf{H}}_{k}$. Assuming the employed channel estimation algorithm is reliable enough such that $|| \hat{\mathbf{H}}_{k}  ||_{\mathsf{F}} \geq \delta_k$, the optimal solution becomes $\boldsymbol{\Delta}_k^* = - \delta_k \hat{\mathbf{H}}_{k} / || \hat{\mathbf{H}}_{k}  ||_{\mathsf{F}}$. Substituting this solution into the objective function~\eqref{eq:apx_of}, we obtain a lower bound on the desired signal component as:
  \begin{equation}
         \mathcal{D}_k(\boldsymbol{\Delta}_k) \geq \Big( 1-\frac{\delta_k}{|| \hat{\mathbf{H}}_{k}  ||_{\mathsf{F}}} \Big)^2 \big| \mathbf{c}_{k,\ell}^\mathsf{H}  \hat{\mathbf{G}}_{k}  (\boldsymbol{\Theta} )  \mathbf{A} \mathbf{d}_{k,\ell } \big|^2.
  \end{equation}
  
  For the interference term, we use the triangle inequality followed by the Cauchy–Schwarz inequality to obtain
\begin{equation}
\begin{split}
    & \big| \mathbf{c}_{k,\ell}^\mathsf{H}  \big( \hat{\mathbf{G}}_k (\boldsymbol{\Theta}) + \boldsymbol{\Lambda}_k \big) \mathbf{A} \mathbf{d}_{i,v } \big| \\ 
    & \leq | \mathbf{c}_{k,\ell}^\mathsf{H}  \hat{\mathbf{G}}_k (\boldsymbol{\Theta}) \mathbf{A} \mathbf{d}_{i,v}| +  |\mathbf{c}_{k,\ell}^\mathsf{H} \boldsymbol{\Lambda}_k \mathbf{A} \mathbf{d}_{i,v}| \\
    & \leq | \mathbf{c}_{k,\ell}^\mathsf{H}  \hat{\mathbf{G}}_k (\boldsymbol{\Theta}) \mathbf{A} \mathbf{d}_{i,v}| + || \mathbf{c}_{k,\ell} ||_2 || \boldsymbol{\Lambda}_k ||_\mathsf{F} || \mathbf{A} \mathbf{d}_{i,v }||_2.
    \end{split}
\end{equation}
Using the facts $|| \boldsymbol{\Lambda}_k ||_\mathsf{F} \leq || \boldsymbol{\Delta}_k ||_\mathsf{F} || \boldsymbol{\theta}_k ||_2 \leq \delta_k \sqrt{N_\mathsf{RIS}}$ and $|| \mathbf{A} \mathbf{d}_{i,v }||_2 = || \mathbf{d}_{i,v }||_2$, we can obtain
 \begin{equation}
 \begin{split}
  | \mathbf{c}_{k,\ell}^\mathsf{H}  \mathbf{G}_k (\boldsymbol{\Theta}) \mathbf{A} \mathbf{d}_{i,v } | \leq &  | \mathbf{c}_{k,\ell}^\mathsf{H}  \hat{\mathbf{G}}_k (\boldsymbol{\Theta}) \mathbf{A} \mathbf{d}_{i,v }| \\
 & + \delta_k \sqrt{N_\mathsf{RIS}} || \mathbf{c}_{k,\ell} ||_2 || \mathbf{d}_{i,v }||_2.
  \end{split}
 \end{equation}

 The noise term can be upper bounded as
 \begin{equation}
 \mathbb{E}\{ |\mathbf{c}_{k,\ell}^\mathsf{H} \mathbf{n}_k|^2 \} \leq || \mathbf{c}_{k,\ell} ||_2^2 \sigma^2.
  \end{equation}
 Thus, \eqref{eq:ny_eq} can be then obtained after dividing all terms by $|| \mathbf{c}_{k,\ell} ||_2^2$.



%





\ifCLASSOPTIONcaptionsoff
  \newpage
\fi





\bibliographystyle{IEEEtran}
\bibliography{IEEEabrv,Bibliography}
%


\end{document}